\documentclass[12pt]{article}

\usepackage[margin=1.25in]{geometry}
\usepackage{setspace}
\usepackage{enumitem}
\usepackage{mathtools}
\usepackage{dsfont} 
\usepackage{amsthm}
\usepackage{amssymb}
\usepackage{breakcites}
\usepackage{dsfont}
\usepackage{bm}
\usepackage{amsfonts}
\usepackage{wasysym} 
\usepackage{graphicx}
\usepackage{float}
\usepackage{epstopdf}
\usepackage[bottom]{footmisc}
\usepackage[titletoc]{appendix}
\usepackage[usenames,dvipsnames]{color} 
\usepackage[colorlinks,citecolor=blue,urlcolor=blue,breaklinks=true]{hyperref}
\usepackage{caption}
\usepackage{subfigure}
\allowdisplaybreaks
\usepackage{csquotes}
\usepackage{wrapfig}
\usepackage[authoryear]{natbib}
\everymath{\displaystyle}

\newtheorem{lemma}{Lemma}
\newtheorem{theorem}{Theorem}
\newtheorem{corollary}{Corollary}
\newtheorem{example}{Example}
\newtheoremstyle{named}{}{}{\itshape}{}{\bfseries}{.}{.5em}{#1\thmnote{ #3}}
\theoremstyle{named}
\newtheorem*{namedassumption}{Assumption}

\newcommand{\K}{\mathcal{K}}
\newcommand{\R}{\mathbb{R}}
\newcommand{\Z}{\mathcal{Z}}

\numberwithin{equation}{section}

\begin{document}

\title{Almost Sure Uniqueness of a Global Minimum\\ 
Without Convexity\footnote{A previous version of this paper was circulated under the title ``Generic Uniqueness of a Global Minimum.'' The author gratefully acknowledges help from conversations with Donald Andrews, Xiaohong Chen, Yuichi Kitamura, Ivana Komunjer, Simon Lee, Adam McCloskey, Jos\'{e} Luis Montiel Olea, Serena Ng, Aureo de Paula, Bernard Salani\'{e}, Tobias Salz, and Ming Yuan.}}
\author{Gregory Cox\footnote{Department of Economics, Columbia University, gfc2106@gmail.com}\\Columbia University\\
February 19, 2019}
\date{}
\maketitle
\vspace{-0.5cm}

\begin{abstract}
This paper establishes the argmin of a random objective function to be unique almost surely. 
This paper first formulates a general result that proves almost sure uniqueness without convexity of the objective function. 
The general result is then applied to a variety of applications in statistics. 
Four applications are discussed, including uniqueness of M-estimators, both classical likelihood and penalized likelihood estimators, and two applications of the argmin theorem, threshold regression and weak identification. 

\mbox{}\\
\mbox{}\\
\mbox{}\\
\mbox{}\\
\mbox{}\\
\mbox{}\\
\mbox{}\\
\mbox{}\\

\textbf{Keywords:} Global Optimization, Nonconvex Optimization, M-estimation, Mixture Model, Argmax Theorem, Threshold Regression, Weak Identification 
\end{abstract}

\pagebreak

\section{Introduction}

This paper establishes the argmin of a random objective function to be unique almost surely. This means that the probability the argmin contains two or more points is zero. The task of finding the argmin of a random function is a very general problem, and evaluating whether the argmin is unique is important in many applications. For example, in M-estimation, uniqueness is equivalent to the estimator being well-defined as a point. In addition, uniqueness is a condition for the argmin theorem, which characterizes the asymptotic distribution of many estimators. 

The usual argument for uniqueness of the argmin relies on convexity. Without convexity, it is difficult to guarantee uniqueness of the argmin. At the same time, there is a popular intuition that multiple global minimizers occurring with positive probability requires a degenerate random function, in some sense. By considering almost sure uniqueness and relying on a type of nondegeneracy condition, this paper provides a systematic way to relax convexity. 

This paper first formulates a general result, Lemma 1, that proves almost sure uniqueness under very weak conditions. This result relies on a type of nondegeneracy condition, called genericity, which states that certain derivatives of the objective function are nonzero. The key to this condition is that it permits derivatives with respect to $z$, a random variable indexing the randomness of the objective function, in addition to derivatives with respect to the domain of optimization. This key aspect makes Lemma 1 useful in a variety of statistical applications. 

At this level of generality, there are not many papers that seek to verify uniqueness of the argmin of a function without convexity. The closest is an approach based on a ``Mountain Pass Lemma.''  This applies if the Hessian of the objective function is positive definite whenever the gradient is zero. The intuition is that between any two minimizers there must exist a local maximum or saddle point. While this condition is sufficient in one dimension, \cite{TaroneGruenhage1975} gives a counterexample in multiple dimensions. A variety of papers, including \cite{MakelainenSchmidtStyan1981}, \cite{Demidenko2008}, and \cite{Mascarenhas2010} supplement the Hessian condition with additional regularity conditions to prove uniqueness of the minimizer. 

This approach has two disadvantages. First, it has narrow scope. The conclusion of the Mountain Pass Lemma is that the local minimizer is unique. Thus, this approach does not work for any function with multiple local minimizers but a unique global minimizer. Second, the Hessian condition can be difficult to verify if the derivatives of the objective function are intractable. In contrast, Lemma 1 applies to functions with multiple local minimizers, and the assumptions of Lemma 1 are easy to verify. 

Lemma 1 has a variety of applications in statistics and, more broadly, optimization. In this paper, we discuss four applications, two examples of M-estimation and two applications of the argmin theorem. See Section 3 for the literature related to each application. 

M-estimators minimize a random objective function. 
The literature proves uniqueness of an M-estimator for a nonconvex objective function only in isolated cases. 
Lemma 1 can be used to prove uniqueness more generally, which we demonstrate in two cases not established in the literature. 
We prove uniqueness for the classical maximum likelihood estimator for a finite mixture model, and we prove uniqueness of the penalized maximum likelihood estimator of a linear regression with a nonconvex penalty. 

The argmin theorem characterizes the asymptotic distribution of an estimator using the argmin of a limiting stochastic process, if it is almost surely unique. 
Lemma 1 can be used to verify the uniqueness condition, which we demonstrate in two cases not established in the literature. 
\cite{MallikBanerjeeSen2013} consider a p-value based estimator of a threshold regression and use the argmin theorem, but are unable to verify the uniqueness condition. 
We verify it using Lemma 1. 
Also, characterizing the asymptotic distribution for estimators of parameters that are weakly identified uses the argmin theorem. 
In this case, we give low-level sufficient conditions for the uniqueness condition to hold, as well as two counterexamples where it does not. 

Section 2 states Lemma 1. Section 3 discusses the applications. Section 4 proves Lemma 1. Section 5 concludes. An appendix contains additional proofs.

\section{A General Uniqueness Lemma}

This paper studies minimizers of an objective function, $t\mapsto Q(t,z)$, where $z$ is random. The following assumption eliminates mass points in the distribution of $z$. 

\begin{namedassumption}[Absolute Continuity]
Let $z$ be an absolutely continuous random $d_z$-vector with distribution $P$. Let $\Z\subset\R^{d_z}$ be a measurable set such that $P(z\in\Z)=1$. 
\end{namedassumption}
Remark: 
\begin{enumerate}
\item The finite dimensionality of $z$ is not restrictive. Infinite dimensional sources of randomness can be accommodated by focusing on a finite dimensional marginal distribution, and conditioning on the remainder. 
Section 3.2.1 demonstrates this in an application in which the randomness is a Gaussian stochastic process. 
\end{enumerate}

The following assumption specifies the domain over which $Q(t,z)$ is minimized. 

\begin{namedassumption}[Manifold]\mbox{} 
Let $T=\cup_{j\in J} T_j$ be a disjoint union of finitely or countably many second-countable Hausdorff manifolds, possibly with boundary or corner.
\end{namedassumption}
Remark: 
\begin{enumerate}
\item Using manifolds, possibly with boundaries or corners, is a flexible way to allow for a variety of shapes to be minimized over. $T_j$ is a manifold with boundary or corner if each point, $t\in T_j$, is locally diffeomorphic to a neighborhood in $\R^{d_{T_j}}_+$, where $\R_+$ denotes the nonnegative real numbers, $d_{T_j}$ denotes the dimension of $T_j$, and where a diffeomorphism is a continuously differentiable function with a continuously differentiable inverse. This definition is slightly more general than common definitions in differential topology because it explicitly accounts for corners and higher-dimensional corners (See how the boundary is handled in \cite{GuilleminPollack1974} 
and especially \cite{Wall2016}, which allows for corners, but not higher-dimensional corners). This generalization is important for many applications in statistics that require irregularly shaped $T$. In the M-estimation application, $T$ is the parameter space. In the weak identification application, $T$ is the identified set, which may have an irregular shape due to bounds. 
\end{enumerate}

For each $t\in T_j$, let $A(t)$ denote a subset of the tangent cone of $t$ to $T_j$. (Formally the tangent cone of $t$ to $T_j$ is defined as the set of derivatives of curves in $T_j$ that start at $t$. Informally the tangent cone indexes directions for which derivatives with respect to $t$ can be defined.) This is a flexible way to accommodate derivatives at boundary points of $t\in T_j$, as well as some nondifferentiability of the objective function, which is permitted by the next assumption. 
\begin{namedassumption}[Continuous Differentiability]\mbox{}
\begin{enumerate}[label=(\alph*)]
\item For each $j\in J$, $Q(t,z)$ is a continuous function of $t\in T_j$ and $z\in\Z$. 
\item For every $z\in\Z$ and $t\in T_j$, $Q(t,z)$ is differentiable with respect to $z$, and the derivative is continuous with respect to $t$ and $z$. 
\item For every $z\in\Z$, $t\in T_j$, and $\Delta\in A(t)$, $Q(t,z)$ is differentiable with respect to $t$ in the direction $\Delta$, and the derivative is continuous with respect to $t$ and $z$. 
\end{enumerate}
\end{namedassumption}
Remark: 
\begin{enumerate}
\item Permitting $A(t)$ to be a strict subset of the tangent cone allows for nondifferentiability of the objective function in directions excluded from $A(t)$. 
\end{enumerate}

Let $\Xi=\{(t,s,z): z\in\Z, t, s\in T, \text{ and } t\neq s\}$. Let $\xi(t,s,z)=Q(t,z)-Q(s,z)$ be defined on $\Xi$. The next assumption is a type of nondegeneracy condition that rules out $t\neq s$ both being global minimizers of $Q(t,z)$ simultaneously. 

\begin{namedassumption}[Generic]
Assume $\xi(t,s,z)$ is a generic function over $\Xi$. That is, for every $(t,s,z)\in \Xi$, at least one of the following is true: 
\begin{enumerate}[label=(\alph*)]
\item $\xi(t,s,z)\neq 0$, 
\item there exists $\Delta\in A(t)$ such that $\frac{d}{d h}\xi(t+h\Delta,s,z)|_{h=0}< 0$, 
\item there exists $\Delta\in A(s)$ such that $\frac{d}{d h}\xi(t,s+h\Delta,z)|_{h=0}> 0$, or 
\item $\frac{d}{dz}\xi(t,s,z)\neq 0$. 
\end{enumerate}
\end{namedassumption}

Remarks: 
\begin{enumerate}
\item The key to Assumption Generic is condition (d). Often, derivatives with respect to $z$ are more tractable than derivatives with respect to $t$. In the applications, the general strategy for verifying Assumption Generic is to show that $\frac{d}{dz}Q(t,z)=\frac{d}{dz}Q(s,z)$ implies $t=s$. 
\item For interior points, $t$ or $s$, conditions (b) and (c) are related to first order conditions for optimality. If the derivative with respect to $t$ is nonzero, then $t$ not a global minimizer, and condition (b) is satisfied for some $\Delta$. These conditions can be augmented with conditions on second derivatives to allow saddle points and local maximizers. 
\item Assumption Generic makes precise the type of degeneracy needed for a random function to have multiple global minimizers with positive probability. Specifically, for interior $t$ and $s$, Assumption Generic is a system of $n+1$ nonlinear equations in $n$ unknowns. Intuitively, if such a system of equations permits a common solution, then it would seem to be degenerate, in some sense. 
\end{enumerate}

\begin{lemma}
Under Assumptions Absolute Continuity, Manifold, Continuous Differentiability, and Generic, the argmin of $Q(t,z)$ over $t\in T$ is unique almost surely-$z$. 
\end{lemma}

Remarks: 
\begin{enumerate}
\item In the special case that $T=\R^m$, the conditions of Lemma 1 can be satisfied without reference to manifolds or tangent cones. Assumption Continuous Differentiability is satisfied if $Q(t,z)$ is continuously differentiable with respect to $(t,z)$. Conditions (b) and (c) in Assumption Generic simplify to 
\begin{enumerate}
\item[(b')] $\frac{d}{dt}\xi(t,s,z)\neq 0$, or 
\item[(c')] $\frac{d}{ds}\xi(t,s,z)\neq 0$. 
\end{enumerate}
Thus, Assumption Generic is satisfied if, for every $(t,s,z)\in\Xi$, $\xi(t,s,z)$ is nonzero, or any of its derivatives are nonzero. 
\item Intuitively, the proof of Lemma 1 eliminates potential global minimizers occurring at distinct points. Each condition in Assumption Generic eliminates $t\neq s$ as simultaneous global minimizers of $Q(t,z)$. Clearly, using the first-order condition, if the derivative with respect to $t$ is nonzero, then we can eliminate that point as a global minimizer. The novel idea is recognizing that we can do the same thing for derivatives with respect to $z$. If the difference of the derivatives at two distinct points, $t\neq s$, with respect to $z$ is nonzero, then the probability of two global minimizers occurring in neighborhoods of those two points is zero. If, with probability 1, all simultaneous global minimizers $t\neq s$ are eliminated, we can conclude that the argmin contains two or more points with probability zero. 
\end{enumerate}

\section{Applications}

\subsection{Nonconvex M-estimation}

Lemma 1 can be applied to estimation methods that minimize a random objective function, also known as M-estimation. In this case, $Q$ is the negative of the likelihood or some other objective function, $T$ is the parameter space, and $z$ is the sample. These optimization problems are known to be nonconvex in general. 

Uniqueness of the argmin is an important property in M-estimation, for a variety of reasons. (1) Although it is not necessary for asymptotic results, such as consistency, uniqueness is the finite sample property that the estimator is a point, a desirable property in itself. (2) In addition, finding the global minimum is a very hard problem numerically, and there are many algorithms, such as multi-start or branch-and-bound, that are designed to find the global minimum. For all of these, uniqueness of the argmin is important for a well-defined convergence criterion. At the same time, uniqueness is a property that is often difficult to verify numerically because the objective function can be very flat or contain many local minimizers in a neighborhood of the global minimizer. (3) Also, uniqueness is important for replication and communication in research. If a replication study calculates a different value of an M-estimator, the study may come to a different conclusion than the original. (4) In addition, \cite{HillierArmstrong1999} provide a formula for the exact density of the maximum likelihood estimator, under the assumption that it is unique, among other regularity conditions. For these reasons, it is useful to have an analytic guarantee that the argmin is unique almost surely. 

\subsubsection{Classical Likelihood}

The canonical example is classical maximum likelihood. A lot of effort has been put into verifying uniqueness of the argmin in isolated cases of nonconvex likelihoods. These examples include the  truncated normal likelihood (\cite{Orme1989} and \cite{OrmeRuud2002}), 
the Cauchy likelihood (\cite{Copas1975}), 
the Weibull likelihood (\cite{ChengChen1988}), 
the Tobit model (\cite{Olsen1978} and \cite{WangBice1997}), 
random coefficient regression models (\cite{Mallet1986}), 
k-monotone densities (\cite{Seregin2010}), 
estimating a covariance matrix with a Kronecker product structure (\cite{RosBijmaMunckGunst2016} and \cite{SoloveychikTrushin2016}), 
and a variety of nonparametric mixture models (\cite{Simar1976}, \cite{HillSaundersLaud1980}, \cite{Lindsay1983a, Lindsay1983b}, \cite{Jewell1982}, \cite{LindsayCloggGrego1991}, \cite{LindsayRoeder1993}, and \cite{Wood1999}). 
All of these examples  require specific knowledge about the structure of the objective function. 

The mixture model is an important example because of its widespread use and the presence of many local minimizers. The cases where uniqueness has been verified, such as in \cite{Lindsay1983a}, are for nonparametric mixture models, where the number of mixture components, $J$, is allowed to be as large as necessary to maximize the likelihood. In some cases, this can be as large as $n/2$, or half of the sample size. This contrasts with finite mixture models, where the number of mixture components is fixed and assumed known. To the author's knowledge, uniqueness of the maximum likelihood estimator has not been verified in finite mixture models. Theorem 1, below, uses Lemma 1 to verify uniqueness of the maximum likelihood estimator for a finite mixture of normal distributions. 

The normal mixture model assumes the sample, $\{z_i\}_{i=1}^n$, is drawn independently and identically from a continuous distribution with density, $f_0(z)$, that is approximated by a mixture density, $f(z;\tau,\mu)$, where $\tau$ is a $J$-vector of weights and $\mu$ is a $J$-vector of means. The mixture density satisfies 
\[
f(z; \tau, \mu)=\sum_{j=1}^J \tau_j \phi(z-\mu_j),
\]
where $\phi(z)$ is the standard normal density. The weights, $\tau_j$, are assumed to be positive and sum to 1. The means, $\mu_j$, are assumed to be strictly increasing: $\mu_1<\mu_2<...<\mu_J$. These assumptions are necessary to ensure that the components can be separately identified. In this case, the parameter to be optimized is $(\tau,\mu)$, while the random vector is the full sample, $z=(z_1,...,z_n)$. We can write the negative of the log-likelihood as: 
\[
Q(\tau,\mu,z)=-\sum_{i=1}^n \log f(z_i;\tau,\mu). 
\]
\begin{theorem} 
If $J\le \sqrt{n}$, then the argmin of $Q(\tau,\mu,z)$ over $(\tau,\mu)$ is unique almost surely-$z$. 
\end{theorem}
Remarks: 
\begin{enumerate}
\item The proof of Theorem 1 verifies condition (d) in Assumption Generic by taking derivatives with respect to $z$ and arguing that $\frac{d}{dz_i}Q(\tau,\mu,z)=\frac{d}{dz_i}Q(\varsigma,\nu,z)$ for all $i=1,...,n$ implies $\tau=\varsigma$ and $\mu=\nu$. 
\item The assumption that $J\le \sqrt{n}$ is an artifact of the proof. The proof needs $n$ to be large so that there are enough derivatives with respect to $z_i$ for the argument in the first remark to be successful. In practice, the assumption that $J\le\sqrt{n}$ is weak. Practical uses of finite mixture models require very few components relative to the sample size. 
\item Theorem 1 does not require the model to be correctly specified. The proof only requires that $z_i$ is continuously distributed. 
\item Theorem 1 demonstrates how Lemma 1 can be used to verify uniqueness of M-estimators. It is stated for a normal mixture, but the proof also covers any mixture of a 1-parameter exponential family. In addition, the proof of Theorem 1 can be extended to any mixture of a $p$-parameter exponential family. 
\item The mixture of a $p$-parameter exponential family includes, as a special case, a normal mixture with unknown variance, but with an important caveat. As a variance parameter goes to zero, the likelihood diverges. Thus, the assumption that $Q(t,z)$ is a real-valued function rules out values of the variance equal to zero. If no lower bound is placed on the variance, no global minimum exists. With a lower bound on the variance, an extension to the proof of Theorem 1 gives almost sure uniqueness. 
\end{enumerate}

Theorem 1 requires $\mu_1<\mu_2<...<\mu_J$. In practice, $Q(\tau,\mu,z)$ is usually minimized without this restriction, resulting in nonuniqueness of the argmin. In this case, Lemma 1 can still be used to characterize the argmin given any one global minimizer. 
When the global minimizer has all distinct means, the argmin is composed of all permutations of the components. An additional complication arises when the global minimizer happens to have multiple components with exactly the same mean. In that case, one can reweight the identical components to find other global minimizers. Corollary 1, below, states this characterization in a way that covers both cases. 

Let $\nu^*$ denote another $J$-vector of means, and let $\varsigma^*$ denote another $J$-vector of positive weights that sum to $1$. 
\begin{corollary}
If $J\le \sqrt{n}$ and if $(\tau^*,\mu^*)$ is a global minimizer of $Q(\tau,\mu,z)$ over the unrestricted parameter space, then the argmin of $Q(\tau,\mu,z)$ over the unrestricted parameter space is equal to 
\[
\{(\varsigma^*,\nu^*): \sum_{j=1}^J\varsigma_j^*\mathds{1}\{\nu^*_j=\mu^*_k\} =\sum_{j=1}^J\tau_j^*\mathds{1}\{\mu^*_j=\mu^*_k\}  \text{ for all } k=1,...,J\} 
\]
almost surely-z. 
\end{corollary}
Remarks: 
\begin{enumerate}
\item Corollary 1 states that the set of all global minimizers can be computed by permuting and reweighting any one global minimizer. 
\item The proof of Corollary 1 demonstrates a general argument for characterizing the argmin of a random objective function, even when the argmin is not unique. The proof transforms the parameter space to a nonredundant version, and then applies Lemma 1. 
\end{enumerate}

\subsubsection{Penalized Likelihood}

Nonconvexity may also arise from a penalty term. Nonconvex penalties are popular because they have desirable properties for recovering sparsity. 
The $L_0$ penalty is the most direct way to impose sparsity. The $L_q$ penalty for $q\in (0,1)$, or bridge penalty, is a continuous penalty that still leads to sparse estimates. \cite{FanLi2001}, \cite{FanPeng2004}, and \cite{FanXueZou2014} 
consider a class of folded concave penalties, including the smoothly clipped absolute deviation (SCAD) penalty. 
\cite{Zhang2010} proposes the minimax concave penalty (MCP), which minimizes the nonconvexity of the penalty subject to constraints. \cite{LohWainwright2017} 
consider a class of nonconvex penalties and give conditions for variable selection consistency. 

In particular, the global minimizer using these penalties has desirable properties. 
\cite{ZhangZhang2012} show that the global minimizer has desirable recovery performance. They also show that the global minimizer is the unique sparse local solution. Also, \cite{HuangHorowitzMa2008} 
show an oracle property for the global minimizer with the bridge penalty, and \cite{KimChoiOh2008} 
show an oracle property for the global minimizer with the SCAD penalty. 

We show that the global minimizer of the penalized likelihood is unique almost surely in the case of the linear regression model with a wide variety of penalties, including all the penalties mentioned above. Let 
\[
Y=X\beta+\epsilon,
\]
where $Y$ is an $n\times 1$ vector and $X$ is a $n\times d$ matrix. We estimate $\beta$ by minimizing 
\[
Q(\beta,Y,X)=\frac{1}{2}\|Y-X\beta\|^2+p(\beta),
\]
where $p(\beta)$ is a penalty term, over $\beta\in B\subset \mathbb{R}^d$. 

\begin{theorem}
Assume $B=\cup_{j=1}^J B_j$, where each $B_j$ is a manifold, possibly with boundary or corner, such that $p(\beta)$ is continuous over each $B_j$. If $X$ is full rank $d$ and the distribution of $Y$ conditional on $X$ is absolutely continuous, then the argmin of $Q(\beta,Y,X)$ over $B$ is unique almost surely. 
\end{theorem}

Remarks: 
\begin{enumerate}
\item The assumption that $X$ is full rank is the same condition for uniqueness in unpenalized least squares. It is surprising that the only additional condition needed for uniqueness in penalized least squares is absolute continuity of $Y$ conditional on $X$. 
\item Theorem 2 accommodates a wide variety of nonconvex penalties. Even some discontinuous penalties can be accommodated, including the $L_0$ penalty by partitioning the parameter space into all possible combinations of $\beta_j=0$ and $\beta_j\neq 0$. 
\item Theorem 2 is stated for the linear regression model, but the argument can be used for more general penalized likelihood models. In fact, if we verify uniqueness of the unpenalized likelihood using Lemma 1, where we verify Assumption Generic using condition (d), as in the finite mixture model, then uniqueness of the global minimizer of the penalized likelihood follows by the same argument for any continuous, deterministic penalty. 
\end{enumerate}

\subsection{The Argmin Theorem}

In many cases, limit theory for M-estimators follows from the argmin theorem (see \cite{KimPollard1990} 
or \cite{VaartWellner1996}). 
An important condition in the argmin theorem is that the limiting stochastic process has a unique minimum almost surely. 

The uniqueness condition has been analyzed in the case that the limiting stochastic process is, itself, a Gaussian process. Papers considering this case include \cite{Lifshits1982}, \cite{KimPollard1990}, \cite{Arcones1992}, \cite{MullerSong1996}, and \cite{Ferger1999}. 
The arguments used in these papers are all specific to proving uniqueness of the minimizer of a Gaussian process, rather than a more general function of a Gaussian process. In addition, \cite{Pimentel2014} and \cite{LopezPimentel2016} characterize uniqueness using differentiability of a perturbation-expectation operator, which is useful in some examples. 

Lemma 1 provides a new technique for verifying the uniqueness condition. In addition to covering the case where the limit is, itself, a Gaussian stochastic process, Lemma 1 is applicable to the more general setting, where the limit is a function of a Gaussian process. We demonstrate the usefulness of Lemma 1 using two novel applications in this setting: p-value based threshold regression and weak identification. 

\subsubsection{Threshold Regression}

Consider the threshold regression model of \cite{MallikBanerjeeSen2013}: 
\[
Y=\mu(X)+\epsilon,
\]
where $\mu(\cdot)$ is a \textit{continuous} function that is equal to a fixed value $\tau$ for $X\le d_0$ and is strictly larger than $\tau$ for $X>d_0$. The parameter of interest is the threshold $d_0$ estimated by a p-value based M-estimator. 

\cite{MallikBanerjeeSen2013} characterize the limit of the objective function as a functional of a Gaussian process. Specifically, let $W(t)$ with $t\in\mathbb{R}$ be a Gaussian process with almost surely continuous sample paths, continuous drift $m(t)$ and continuous covariance kernel $\Sigma_{t_1,t_2}$. 
The limiting objective function is 
\begin{equation}
Q(t,W(\cdot))=\int_0^t \Phi(W(y))dy-t\gamma,
\end{equation}
where $\Phi(\cdot)$ is the standard normal cdf and $\gamma$ is a constant. This defines a functional of a Gaussian process. \cite{MallikBanerjeeSen2013} 
are unable to prove that (3.1) has a unique minimum almost surely, but assume uniqueness in order to invoke the argmin theorem. We show that under the assumption $\Sigma_{t,s}>0$ for all $t, s\in\mathbb{R}$ (which follows from Assumption 3(a) and Lemma 2 in \cite{MallikBanerjeeSen2013}), 
the minimum is indeed almost surely unique. 
\begin{theorem}
If $\Sigma_{t,s}>0$ for all $s,t\in\mathbb{R}$, then the argmin of $Q(t,W(\cdot))$, defined in equation (3.1), over $t\in\mathbb{R}$ is unique almost surely. 
\end{theorem}
Remark: 
\begin{enumerate}
\item The proof of Theorem 3 demonstrates how Lemma 1 can accommodate an infinite dimensional source of randomness: by taking derivatives with respect to $Z=W(M)$, for fixed values of $M$. 
\end{enumerate}

\subsubsection{Limit Theory for Weakly Identified Parameters}

Limit theory for estimators of weakly identified parameters relies on the argmin theorem. This requires a unique minimum assumption on the limit of the profiled objective function. Papers that use this assumption include \cite{StockWright2000}, \cite{AndrewsCheng2012}, \cite{Cheng2015}, \cite{Cox2019}, and \cite{HanMcCloskey2018}. 
\cite{AndrewsCheng2012} provide sufficient conditions for uniqueness in the special case that the key parameter, which determines the strength of identification, is scalar. However, examples that require a vector of key parameters, including \cite{Cheng2015} and \cite{Cox2019}, 
can benefit from the low-level sufficient conditions stated in this paper. 

Following \cite{Cox2019}, one first reparameterizes the model into the identified parameters, $\beta$, and the unidentified parameters, $\pi$. Next, one defines a function, $h: (\beta,\pi)\mapsto h(\beta,\pi)\in\mathbb{R}^{d_h}$, that maps structural parameters to identified reduced-form parameters. $h$ has the property that for some values of $\beta$, $h$ is injective as a function of $\pi$, and then $\pi$ is identified, but for other values of $\beta$, $h$ is not injective as a function of $\pi$, and then $\pi$ is not identified. In this way, the identifiability of $\pi$ depends on the true value of $\beta$. The simplest example of a function satisfying this property is $h(\beta,\pi)=\beta\pi$, which is injective as a function of $\pi$ if and only if $\beta\neq 0$. 

Consider an estimator $(\hat\beta,\hat\pi)$ that minimizes a random objective, $Q_n(\beta,\pi)$. To derive the asymptotic distribution of $\hat\pi$ we consider the profiled objective, $Q_n^p(\pi)=\min_\beta Q_n(\beta,\pi)$. Appropriately standardized, this converges to a limiting stochastic process over $\pi$ whose argmin characterizes the asymptotic distribution, if it is unique. In what follows we give the formula for the limit, as well as sufficient conditions for the argmin to be unique. 

Let parameters $\beta$ and $\pi$ have dimensions $d_\beta$ and $d_\pi$, respectively. We characterize the limit along a sequence of true values of the parameters, $\beta_n$ converging to $\beta_0$, a point for which $h$ is not an injective function of $\pi$. These sequences lead to an intermediate identification strength, called weak identification, indexed by a local parameter $b\in\R^{d_\beta}$. The asymptotic distributions are continuous in this local parameter, and thus are the appropriate sequences for contiguity. The case $b=0$ is an important special case corresponding to a complete loss of identification. It derives from $\beta_n=\beta_0$ for all $n$. A typical sequence satisfies the following assumption, which says that $\beta_n$ influences the value of $h(\beta,\pi)$ at the $\sqrt{n}$ rate. 
\begin{namedassumption}[Weak Identification]\mbox{}
\begin{enumerate}[label=(\alph*)]
\item $h(\beta,\pi)$ is twice continuously differentiable, and 
\item for some $b\in\mathbb{R}^{d_\beta}$, $\sqrt{n}\left[h(\beta_n,\pi)-h(\beta_0,\pi)\right]\rightarrow h_\beta(\beta_0,\pi)b$, uniformly on compact sets over $\pi$, where $h_\beta(\beta,\pi)$ denotes the derivative of $h(\beta,\pi)$ with respect to $\beta$.  
\end{enumerate}
\end{namedassumption}

In this application, the parameter $\pi$ serves the same purpose as $t$, indexing the domain of the random function. The domain is the identified set for $\pi$ under identification loss. Allowing the domain to be a union of manifolds is useful because the identified set often has an unusual shape that may be difficult to characterize exactly. Following calculations in \cite{Cox2019}, there exists a continuous random vector, $z$, with dimension $d_z=d_\beta+d_h$, and there exists a symmetric and positive definite $d_z\times d_z$ matrix, $H$, such that the limit of the profiled objective function is 
\begin{equation}
Q(\pi,z)=2z'{H^{1/2}}'g(\pi)-(H^{1/2}z+g(\pi))'P(\pi)(H^{1/2}z+g(\pi))+\kappa(\pi), 
\end{equation}
where 
\begin{align*}
g(\pi)&=H^{1/2}\left[\begin{array}{c}0_{d_\beta\times 1}\\ h_\beta(\beta_0,\pi)b-h_\beta(\beta_0,\pi_0)b\end{array}\right],\\
S(\pi)&=H^{1/2}\left[\begin{array}{c}I_{d_\beta}\\ h_\beta(\beta_0,\pi)\end{array}\right],\\
P(\pi)&=S(\pi)[S(\pi)'S(\pi)]^{-1}S(\pi)', 
\end{align*}
and $\kappa(\pi)$ is some continuous deterministic function of $\pi$. 
Notice that this limit is indexed by a finite dimensional random vector, $z$, rather than an infinite dimensional stochastic process, and hence is more susceptible to nonuniqueness. 

The following theorem places low-level conditions on $h$ in order to show that the argmin of $Q(\pi,z)$ over $\pi$ is almost surely unique.
\begin{theorem}
Assume the identified set can be written as a finite or countable disjoint union of second-countable Hausdorff manifolds. Let $h(\beta,\pi)$ and the sequence $\beta_n$ satisfy Assumption Weak Identification. Let $Q(\pi,z)$ be defined in equation (3.2). If  
\begin{enumerate}[label=(\alph*)]
\setcounter{enumi}{2}
\item for all $\pi_1\neq \pi_2$, the rank of $h_\beta(\beta_0,\pi_1)-h_\beta(\beta_0,\pi_2)$ is $d_h$, and 
\item there exists an open set $B$ containing $\beta_0$, such that for almost every $\beta\in B$, $h(\beta,\pi)-h(\beta_0,\pi)$ is an injective function of $\pi$. 
\end{enumerate}
Then, $\xi(\pi_1,\pi_2,z)=Q(\pi_1,z)-Q(\pi_2,z)$ satisfies Assumption Generic. Therefore, by Lemma 1, the argmin of $Q(\pi,z)$ over $\pi$ is unique almost surely. 
\end{theorem}
Remarks: 
\begin{enumerate}
\item Conditions (c) and (d) eliminate degeneracy in $Q(\pi,z)$ as a function of $\pi$ so that Assumption Generic can be verified by taking derivatives with respect to $z$. Condition (c) is a rank condition guaranteeing that $h_\beta(\beta_0,\pi)$ varies enough as a function of $\pi$. A necessary condition is that $d_\pi\le d_h$. 
Condition (d) says that $\pi$ is generically identified locally around $\beta_0$. Below, two examples are given that demonstrate the importance of these two conditions. 
\item The proof of Theorem 4 using Lemma 1 is nontrivial and requires an appeal to Sard's theorem to characterize the critical values of $h_\beta(\beta_0,\pi)(z_1-b)$ as a function of $\pi$. 
\end{enumerate}

\begin{figure}
\begin{center}
\begin{tabular}{cc}
\includegraphics[scale=0.5]{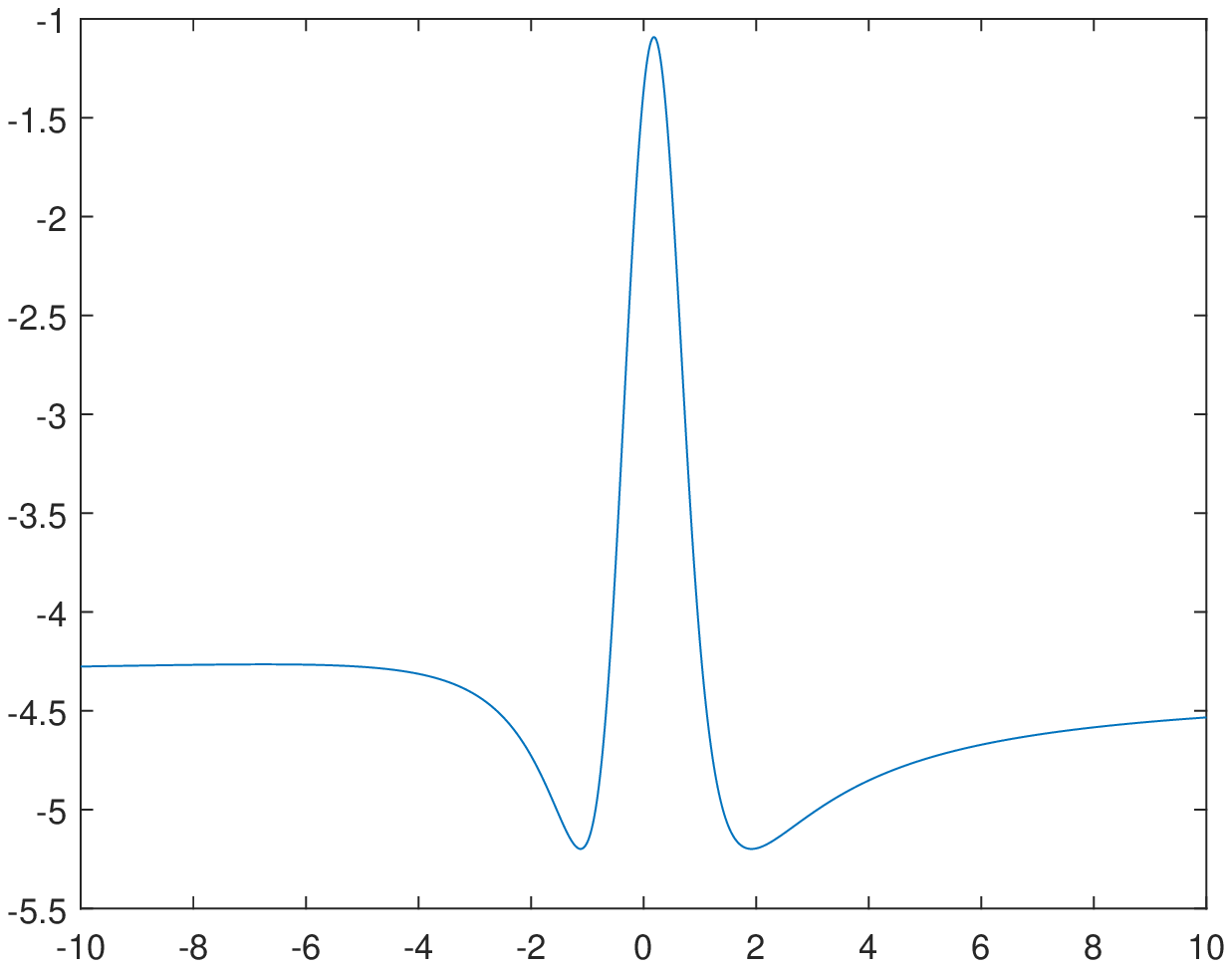}&\includegraphics[scale=0.5]{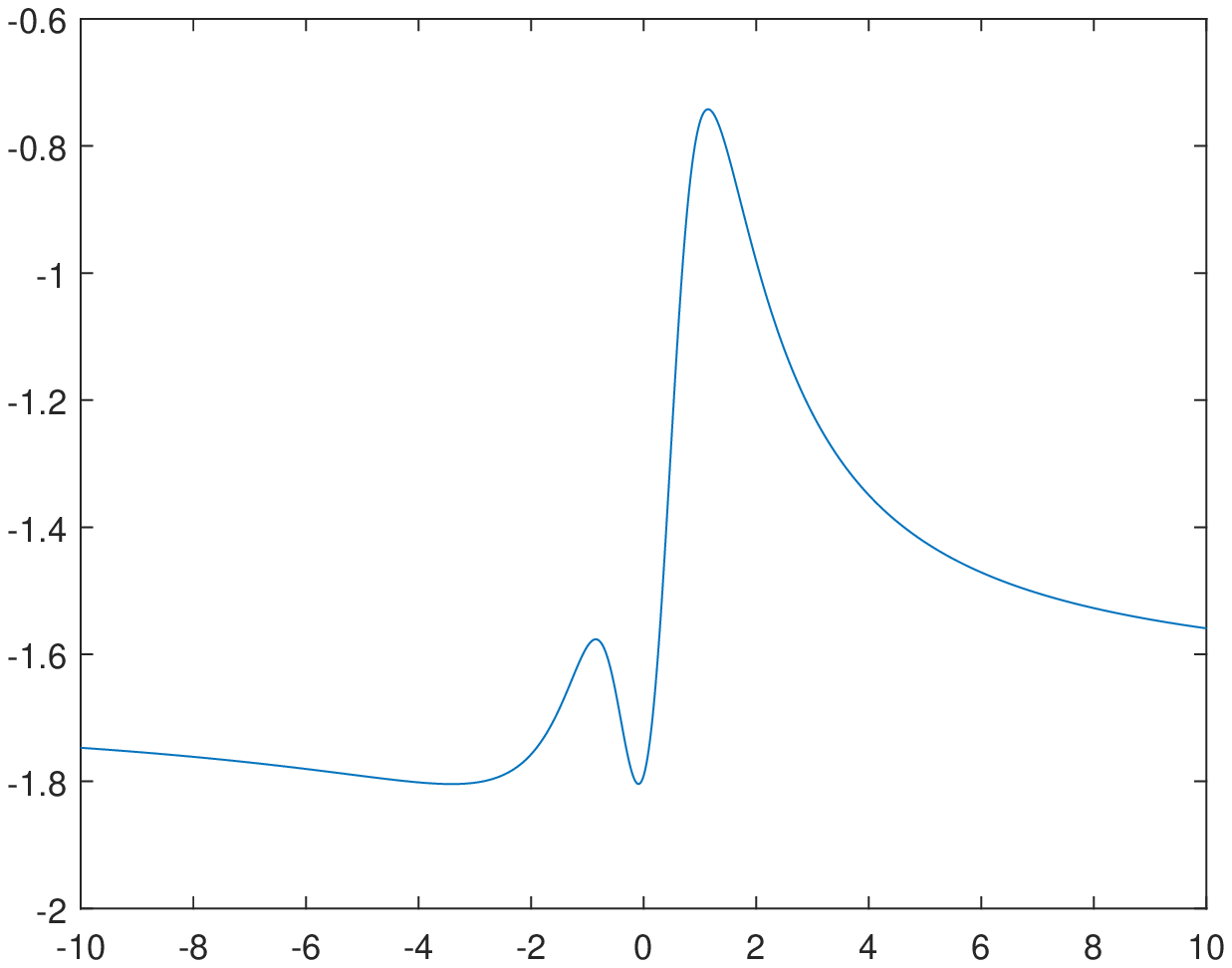}
\end{tabular}
\end{center}
\caption{Sample simulations for $Q(\pi,z)$ from Example 1 when $b=(0,0)'$ and $H=I_3$. In the left panel, $z=(-1.03,1.29,2.77)'$ and in the right panel, $z=(-1.82,-0.52,0.16)'$. These values were chosen randomly and are representative of a positive probability of multiple global minimizers.}
\end{figure}

\begin{example}
This example demonstrates the importance of condition (d), that $\pi$ is generically identified in a neighborhood of $\beta_0$. Consider the model, 
\[
Y_i=\alpha+\beta_1 X_{1i}+\beta_2 X_{2i}+(\beta_1\pi+\beta_2\pi^2)X_{3i}+u_i, 
\]
where $\mathbb{E}(u_i|X_{1i},X_{2i},X_{3i})=0$. In this case, identification of $\pi$ is determined by injectivity of $h(\beta_1,\beta_2,\pi)=\beta_1\pi+\beta_2\pi^2$ in a neighborhood of $(\beta_1,\beta_2)=(0,0)$. Condition (d) is not satisfied because, for any $\beta_1$ and $\beta_2\neq 0$, there exists an $h\in\mathbb{R}$ such that the quadratic equation, $\beta_1\pi+\beta_2\pi^2=h$, has multiple solutions in $\pi$. We can calculate $\frac{d}{d\beta}h(\beta_1,\beta_2,\pi)|_{\beta_1=0, \beta_2=0}=[\pi, \pi^2]$, which satisfies condition (c). If $\pi$ is estimated by nonlinear least squares, the limit of the profiled objective function, $Q(\pi,z)$, has multiple minimizers with positive probability. Figure 1 gives some simulations of this function. 
\end{example}

\begin{figure}
\begin{center}
\begin{tabular}{cc}
\includegraphics[scale=0.5]{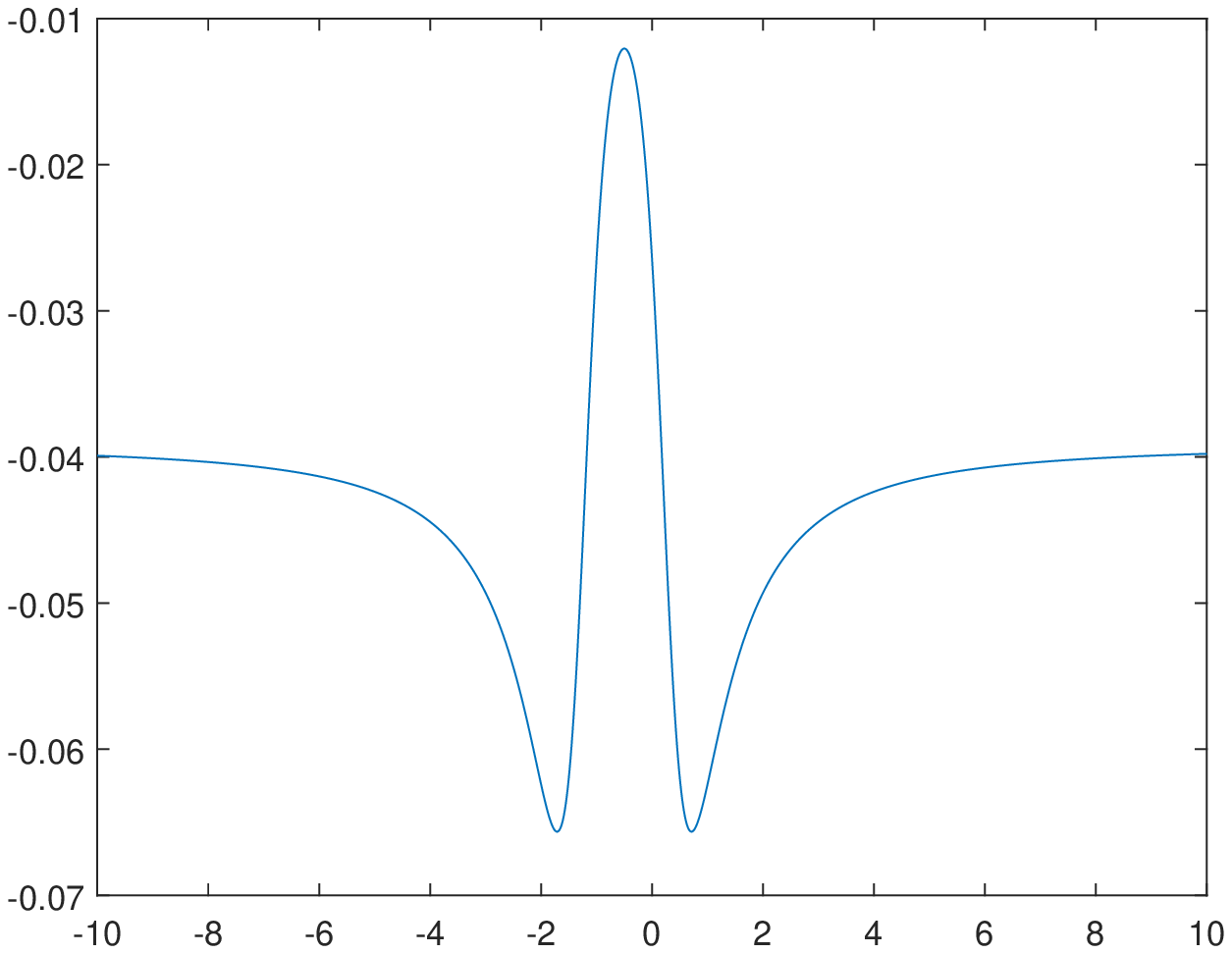}&\includegraphics[scale=0.5]{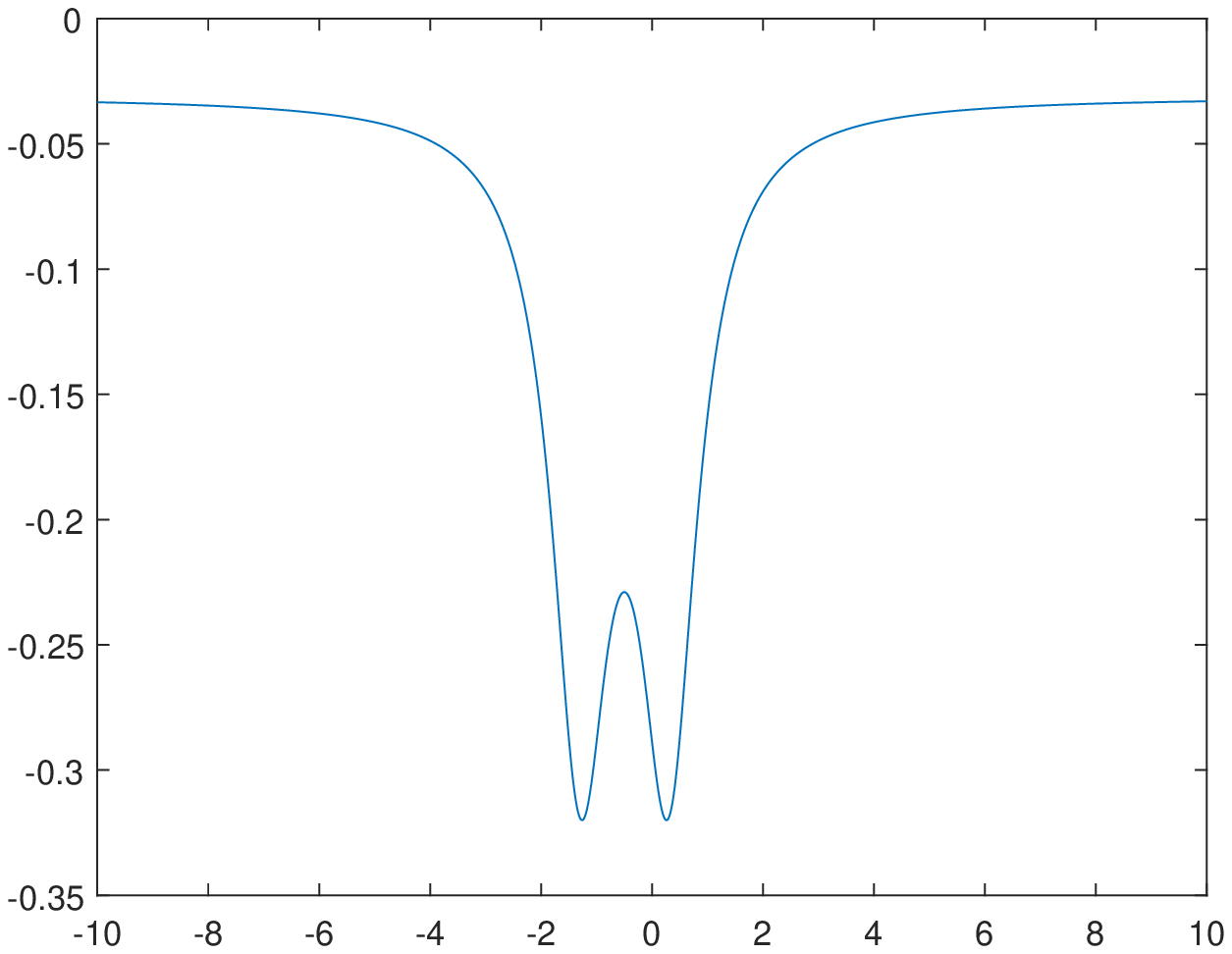}
\end{tabular}
\end{center}
\caption{Sample simulations for $Q(\pi,z)$ from Example 2 when $b=0$ and $H=I_3$. In the left panel, $z=(-0.23,-0.28,1.31)'$ and in the right panel, $z=(-0.76,-0.25,-1.65)'$. These values were chosen randomly and are representative of a positive probability of multiple global minimizers.}
\end{figure}

\begin{example}
This example demonstrates the importance of condition (c), which states that $h_\beta(\beta_0,\pi_1)-h_\beta(\beta_0,\pi_2)$ has full rank, $d_h$, for all $\pi_1\neq \pi_2$. Consider the model, 
\[
Y_i=\alpha+\beta X_{1i}+\beta(\pi+\pi^2) X_{2i}+\beta^2\pi X_{3i}+u_i, 
\]
where $\mathbb{E}(u_i|X_{1i},X_{2i},X_{3i})=0$. In this case, identification of $\pi$ is determined by injectivity of 
\[
h(\beta,\pi)=\left[\begin{array}{c}\beta(\pi+\pi^2)\\ \beta^2\pi\end{array}\right], 
\]
in a neighborhood of $\beta=0$. Condition (d) is satisfied because for any $\beta\neq 0$, $\beta^2\pi$ is injective as a function of $\pi$. We can calculate 
\[
\frac{d}{d\beta}h(\beta,\pi)|_{\beta=0}=\left[\begin{array}{c}\pi+\pi^2\\ 0 \end{array}\right],
\]
which does not satisfy condition (c) because the rank is zero whenever $\pi_1$ and $\pi_2$ are both roots of $\pi+\pi^2=c$. 

If $\pi$ is estimated by nonlinear least squares, the limit of the profiled objective function, $Q(\pi,z)$, has multiple minimizers with positive probability. Figure 2 gives some simulations of this function. The key components of this example are the two functions in $h(\beta,\pi)$ that depend nonlinearly on $\beta$, and contain different amounts of information about $\pi$. As $\beta_n\rightarrow 0$, the more informative function is weaker, and therefore cannot satisfy condition (c). This example is concerning because it seems likely that these key components are present in more complicated weakly identified models. 
\end{example}

\section{The Proof of Lemma 1}

This section gives the proof of Lemma 1. The proof is simple and intuitive. The first subsection reduces the global problem to a local problem, the second subsection states lemmas for the local problem, and the third subsection finishes the proof. The proofs of the additional lemmas are in the appendix. 

\subsection{Global to Local}

We first reduce from minimization over all of $T$ to minimization over a countable collection of compact neighborhoods, $K\in \K_j$, that separate points of $T_j$. Since each $T_j$ is second-countable and Hausdorff, this countable collection always exists. 
For any $K\subset T$, define the value function, 
\[
V(K,z)=\inf_{t\in K} Q(t,z). 
\]

We consider $j_1,j_2\in J$. Let $K\in\mathcal{K}_{j_1}$ and $C\in\mathcal{K}_{j_2}$ be disjoint. Lemma 2 shows that the value of the minimum over disjoint compact sets, $K$ and $C$, being different from each other or from the value of the minimum over all of $T$ is sufficient for the argmin to be unique almost surely. 
\begin{lemma}
Suppose that $\{\Z_k\}_{k=1}^{\infty}$ is a sequence of compact subsets of $\Z$ such that $P(\Z_k)\rightarrow 1$ as $k\rightarrow\infty$, and suppose that for every $k\in\mathbb{N}$, for every $j_1,j_2\in J$, for every $K\in \mathcal{K}_{j_1}$ and $C\in\mathcal{K}_{j_2}$ such that $K\cap C=\emptyset$,  
\[
P\left(\{z\in \Z_k| V(T,z)=V(K,z)=V(C,z)\}\right)=0.
\]
Then, the argmin of $Q(t,z)$ over $t\in T$ is unique almost surely. 
\end{lemma}

The condition in Lemma 2 is still not local. Lemma 3 reduces this condition to a local condition by finding neighborhoods of $t$, $s$, and $z$ that can be used to cover these compact sets such that an appropriate probability is zero. 

\begin{lemma}
Fix $\Z_k\subset \Z$, compact, fix $j_1,j_2\in J$, and fix $K\in\mathcal{K}_{j_1}$ and $C\in\mathcal{K}_{j_2}$ such that $K\cap C=\emptyset$. Suppose, for every $\bar z\in\Z_k$ and for every $(t,s)\in K\times C$, there exist neighborhoods, $N_{t, s,\bar z}, M_{t,s,\bar z}$, and $W_{t,s,\bar z}$ of $t, s$, and $\bar z$, respectively, such that 
\begin{align}
P\left(\{z\in W_{t,s,\bar z}| \right.V(T,z)=V(K,z)&=V(N_{t,s,\bar z},z)\nonumber\\
&\left.=V(C,z)=V(M_{t,s,\bar z},z)\}\right)=0.
\end{align}
Then, 
\[
P\left(\{z\in\Z_k| V(T,z)=V(K,z)=V(C,z)\}\right)=0.
\]
\end{lemma}

\subsection{The Local Problem}

For the rest of this section, fix $j_1,j_2\in J$, $K\in \mathcal{K}_{j_1}$, $C\in\mathcal{K}_{j_2}$ such that $K\cap C=\emptyset$, $t\in K$, $s\in C$, and $\bar{z}\in \Z_k$. We state lemmas that are useful for showing the existence of neighborhoods, $N, M$, and $W$, that satisfy (4.1), using properties of $Q(t,\bar z)$ that follow from Assumption Generic. Assumption Generic implies one of three conditions: 
\begin{enumerate}
\item $Q(t,\bar z)\neq Q(s,\bar z)$, 
\item there exists a $\Delta\in A(t)$ such that $\frac{d}{dh}Q(t+h\Delta,\bar z)|_{h=0}<0$, 
and symmetrically for $s$, and 
\item $\frac{d}{dz}Q(t,\bar z)\neq\frac{d}{dz}Q(s,\bar z)$. 
\end{enumerate}
Lemmas 4, 5, and 6, below, show the existence of neighborhoods that satisfy (4.1) for each of these cases, respectively. 

Lemma 4 follows from the intuitive notion that if the value of $Q(t,\bar z)$ is far from the value of $Q(s,\bar z)$, then the value of $V(N,z)$ is far from the value of $V(M,z)$, for small enough neighborhoods of $t$, $s$, and $\bar z$. 
\begin{lemma}
If $Q(t,\bar z)\neq Q(s,\bar z)$, then there exist neighborhoods, \hspace{-0.2mm}$N$\hspace{-0.4mm}, \hspace{-0.2mm}$M$\hspace{-0.4mm},\hspace{-0.3mm} and $W$\hspace{-0.4mm}, so that 
\[
\{z\in W | V(N,z)=V(M,z)\}=\emptyset.
\]
\end{lemma}

Lemma 5 uses the first order conditions for optimality of $t$ or $s$. 
It uses the intuitive notion that if $t$ is not a relative minimum of $Q(t,\bar z)$, then it is not a global minimum of $Q(t,\bar z)$ over $T$, and it can be bounded away from the minimum in a neighborhood of $\bar z$. 
\begin{lemma}
If there exists a $\Delta\in A(t)$ such that 
\[
\frac{d}{d h}\left.Q(t+h \Delta,\bar z)\right|_{h=0}<0,
\]
then there exist neighborhoods of $t$ and $\bar z$, $N$ and $W$\hspace{-0.5mm}, such that 
\[
\{z\in W| V(N,z)=V(T,z)\}=\emptyset.
\] 
\end{lemma}

Lemma 6 is a novel contribution allowing a simple proof of Lemma 1. 
It relies on a type of mean value bound for secants of the value function. 
The intuition is if, in some $z$-direction, the derivative of $Q(t,\bar z)$ is always less than the derivative of $Q(s,\bar z)$, then any secants of $V(N,\bar z)$ and $V(M,\bar z)$ share that property, for sufficiently small neighborhoods, $N$ and $M$. 
Thus, $V(N,\bar z)$ is always increasing or decreasing at a rate which is less than the rate at which $V(M,\bar z)$ is increasing or decreasing. 
This implies that they cannot cross more than once, and the set of crossing points must have probability zero. 
\begin{lemma}
If 
\[
\frac{d}{dz}Q(t,\bar z)\neq \frac{d}{dz}Q(s,\bar z),
\]
then there exist neighborhoods, $N$, $M$, and $W$, so that 
\[
P(\{z\in W| V(N,z)=V(M,z)\})=0.
\]
\end{lemma}

Lemmas 4-6 are sufficient to find neighborhoods that satisfy (4.1). We now use them to prove Lemma 1. 

\subsection{Proof of Lemma 1}

For any $P$ over $\R^{d_z}$, there exists a sequence of compact sets, $\Z_k\subset \Z$, such that $P(\Z_k)\rightarrow 1$ as $k\rightarrow\infty$. Fix $j_1, j_2\in J$, $K\in\mathcal{K}_{j_1}$, and $C\in\mathcal{K}_{j_2}$ such that $K\cap C=\emptyset$. We seek to verify the conditions of Lemma 3. Fix $t\in K$, $s\in C$, and let $\bar z\in\Z_k$. 
We divide into cases: 
\begin{enumerate}
\item $Q(t,\bar z)\neq Q(s,\bar z)$. In this case, the existence of neighborhoods, $N$, $M$, and $W$, satisfying (4.1) follows from Lemma 4. 
\item There exists a $\Delta\in A(t)$ such that $\frac{d}{dh} Q(t+h\Delta,\bar z)|_{h=0}<0$ or $\frac{d}{dh} Q(s+h\Delta,\bar z)|_{h=0}<0$. In both cases, the existence of neighborhoods satisfying equation (4.1) follows from Lemma 5. 
\item $\frac{d}{d z}Q(t,\bar z)\neq \frac{d}{d z}Q(s,\bar z)$. In this case, the existence of neighborhoods, $N$, $M$, and $W$, satisfying (4.1) follows from Lemma 6.
\end{enumerate}

The above cases exhaust the possibilities. Thus, for every $(t,s)\in K\times C$, the condition of Lemma 3 is satisfied for $K$ and $C$. Since $K$ and $C$ are arbitrary, this verifies the condition of Lemma 2. Therefore, by Lemma 2, the argmin of $Q(t,z)$ over $t\in T$ is unique almost surely-$z$. \qed

\section{Conclusion}

This paper establishes the argmin of a random objective function to be unique almost surely. 
This paper first formulates a general result, Lemma 1, that proves uniqueness without convexity of the objective function. 
This paper applies the result to prove uniqueness in a variety of statistical applications. 
In M-estimation, uniqueness of the argmin is established for the first time in finite mixture models and penalized linear regression with nonconvex penalty. 
In the argmin theorem, uniqueness of the argmin of the limiting stochastic process is established for the first time in two cases: p-value based threshold estimation and the profiled objective function for weakly identified parameters. 

\appendix

\section{Proofs of Lemmas 2-6}

\begin{proof}[\bf{Proof of Lemma 2}]

Let $z\in \Z_k$ for some $k$. Suppose $Q(t,z)$ is not uniquely minimized over $t\in T$. Then, there exist $j_1, j_2\in J$, $t\in T_{j_1}$, and $s\in T_{j_2}$, $t\neq s$, so that $Q(t,z)=Q(s,z)=\inf_{t\in T} Q(t,z).$ Furthermore there exist sets, $K\in\mathcal{K}_{j_1}$ and $C\in \mathcal{K}_{j_2}$, such that $K\cap C=\emptyset$, $t\in K$, and $s\in C$. It follows that $V(T,z)=V(K,z)=Q(t,z)=V(C,z)=Q(s,z).$ This implies that for every $k$, 
\begin{eqnarray*}
&&\{z \in \Z_k|Q(t,z) \text{ is not uniquely minimized over } t\in T\}\\
&\subset&\hspace{-3mm}\underset{j_1,j_2\in J}{\cup}\underset{\underset{K\cap C=\emptyset}{K\in\mathcal{K}_{j_1}, C\in\mathcal{K}_{j_2}}}{\cup}\left\{z\in \Z_k|V(T,z)=V(K,z)=V(C,z)\right\}. 
\end{eqnarray*}
This implies, by countability of $J$ and $\mathcal{K}_j$, that
\begin{eqnarray*}
&&P(\{z\in \Z|Q(t,z) \text{ is not uniquely minimized over } t\in T\})\\
&\le&P(\Z_k^c)+P(\{z\in \Z_k|Q(t,z) \text{ is not uniquely minimized over } t\in T\})\\
&\le&P(\Z_k^c)+\sum_{j_1, j_2\in J}\sum_{\underset{K\cap C=\emptyset}{K\in\mathcal{K}_{j_1}, C\in\mathcal{K}_{j_2}}}P\left(\left\{z\in \Z_k |V(T,z)=V(K,z)=V(C,z)\right\}\right)\\
&=&P(\Z_k^c)\rightarrow 0,
\end{eqnarray*}
where $\Z_k^c$ denotes the complement of $\Z_k$ in $\Z$, the equality follows by assumption, and the convergence follows as $k\rightarrow\infty$ by the assumption on $\Z_k$. 
\end{proof}

\begin{proof}[\bf{Proof of Lemma 3}]

Notice that $\{N_{t,s,\bar z}\times M_{t,s,\bar z}\times W_{t,s,\bar z}|\bar z\in\Z_k, t\in K,\text{ and }s\in C\}$ is an open cover of $\Z_k\times K\times C$. Thus, there is a finite subcover that we index by $\{(t_m, s_m, \bar z_m)\}_{m=1}^{M}.$ Then, 
\begin{eqnarray}
&&P\left(\{z\in\Z_k|V(T,z)=V(K,z)=V(C,z)\}\right)\nonumber\\
&\le&\sum_{m=1}^{M}P\left(\{z\in W_{t_m,s_m,\bar z_m}|V(T,z)=V(N_{t_m,s_m,\bar z_m},z)\right.\\
&&\hspace{4cm}\left.=V(M_{t_m,s_m,\bar z_m},z)=V(K,z)=V(C,z)\}\right)\nonumber\\
&=&0,
\end{eqnarray}
where (A.1) follows from the following argument. Let $z\in \Z_k$. By compactness there exist, $\tilde t\in K$ and $\tilde s\in C$ such that $Q(\tilde t,z)=\inf_{t\in K}Q(t,z)$ and $Q(\tilde s,z)=\inf_{t\in C}Q(t,z)$. Using the open cover, there exists an $m\in\{1,...,M\}$ such that $N_{t_m,s_m,\bar z_m}\times M_{t_m,s_m,\bar z_m}\times W_{t_m,s_m,\bar z_m}$ contains $(\tilde{t},\tilde{s},z)$. This implies that $V(K,z)=V(N_{t_m,s_m,\bar z_m},z)$ and  $V(C,z)=V(M_{t_m,s_m,\bar z_m},z)$. 

Equation (A.2) follows by assumption. 
\end{proof}

\begin{proof}[\bf{Proof of Lemma 4}]

There exists an $\epsilon>0$ such that $|Q(t,\bar z)-Q(s,\bar z)|>\epsilon.$ By continuity of $Q(t,z)$, there exist $W$, $N$, and $M$, bounded neighborhoods of $\bar{z}$, $t$, and $s$, respectively, such that 
\begin{equation}
\inf_{\tilde t\in N}\inf_{\tilde s\in M}\inf_{z\in W}|Q(\tilde t,z)-Q(\tilde s,z)|>\epsilon.
\end{equation}
These neighborhoods satisfy Lemma 4. To see this, fix $z\in W$ and let $\bar t\in cl(N)$ and $\bar s\in cl(M)$ satisfy $Q(\bar t,z)=V(N,z)$ and $Q(\bar s,z)=V(M,z)$, where $cl(\cdot)$ denotes closure. Then, by continuity of $Q(t,z)$, $|Q(\bar t,z)-Q(\bar s,z)|\ge\epsilon$, which implies that $V(N,z)\neq V(M,z)$. 
\end{proof}

\begin{proof}[\bf{Proof of Lemma 5}]

Since $\frac{d}{dh}\left.Q(t+h\Delta,\bar z)\right|_{h=0}<0$, there exists an $\epsilon>0$ and a $\bar t\in T$ such that $Q(\bar t,\bar z)<Q(t,\bar z)-\epsilon.$ By the continuity of $Q(t,\bar z)$, there exist neighborhoods of $\bar z$ and $t$, $W$ and $N$, respectively, such that for all $t^\dagger\in N$ and for all $z\in W$, $Q(\bar t,z)<Q(t^\dagger,z)-\epsilon.$ This implies that for every $z\in W$, $V(N,z)\ge Q(\bar t,z)+\epsilon > V(T,z),$ where the second inequality follows from the fact that $\epsilon>0$. This shows that $\{z\in W | V(N,z)= V(T,z)\}=\emptyset$. 
\end{proof}

\begin{proof}[\bf{Proof of Lemma 6}]

There exists a $\lambda\in\R^{d_z}$ and an $\epsilon>0$ such that 
\[
\frac{d}{dh}\left(Q(t,\bar z+h\lambda)-Q(s,\bar z+h\lambda)\right)|_{h=0}<-\epsilon. 
\]
By Assumption Continuous Differentiability, there exists a convex neighborhood of $\bar{z}$, $W$, and neighborhoods $N$ and $M$ of $t$ and $s$, respectively, such that 
\begin{equation}
\sup_{z\in W}\frac{d}{dh}\left.\left(Q(\tilde t,z+h\lambda)-Q(\tilde s,z+h\lambda)\right)\right|_{h=0}<-\epsilon
\end{equation}
for all $\tilde t\in N$ and $\tilde s\in M$. 

Without loss of generality, we rotate $z$ so that $\lambda=e_1$, the first unit vector. Then, 
\begin{eqnarray*}
&&P(\{z\in W|V(N,z)=V(M,z)\})\\
&=&\int_{W} \mathds{1}\{V(N,z)=V(M,z)\} p(z)dz\\
&=&\int_{\R^{d_z-1}}\int_{\R} \mathds{1}\{V(N,(z_1,z_2))=V(M,(z_1,z_2))\}\mathds{1}\{(z_1,z_2)\in W\}p(z)dz_1dz_2,
\end{eqnarray*}
where the first equality uses Assumption Absolute Continuity with $p(z)$, the density of $P$, and the second equality uses Tonelli's theorem. Thus, it is sufficient to show that for every $z_2$ fixed, the set of all $z_1\in\R$ such that $V(N,z_1,z_2)=V(M,z_1,z_2)$ and $(z_1,z_2)\in W$ is finite. 

We show that the number of such $z_1$ is at most one. Suppose there exist two, $z'_1>z''_1$. Let $z'=(z'_1,z_2)$ and $z''=(z''_1,z_2)$. Then, 
\begin{eqnarray*}
0&\hspace{-0.5mm}=&\hspace{-0.5mm}\frac{V(N,z')-V(N,z'')}{z'_1-z''_1}-\frac{V(M,z')-V(M,z'')}{z'_1-z''_1}\\
&\hspace{-0.5mm}=&\hspace{-0.5mm}\lim_{m\rightarrow\infty} \frac{Q(t_m^* (z') ,z') - Q(t_m^* (z'') ,z'')}{z'_1-z''_1} - \frac{Q(s_m^* (z') ,z') - Q(s_m^* (z'') ,z'')}{z'_1-z''_1}\\
&\hspace{-0.5mm}=&\hspace{-0.5mm}\lim_{m\rightarrow\infty}\frac{Q(t_m^* (z') ,z') - Q(s_m^* (z') ,z') - \left(Q(t_m^* (z'') ,z'') - Q(s_m^* (z'') ,z'') \right)}{z'_1-z''_1}\\
&\hspace{-0.5mm}\le&\hspace{-0.5mm}\overset{\text{limsup}}{\scriptstyle{m\rightarrow\infty}}\hspace{2mm}\frac{Q(t_m^* (z'') ,z') - Q(s_m^* (z') ,z') - \left(Q(t_m^* (z'') ,z'') - Q(s_m^* (z') ,z'') \right)}{z'_1-z''_1}\\
&\hspace{-0.5mm}=&\hspace{-0.5mm}\overset{\text{limsup}}{\scriptstyle{m\rightarrow\infty}}\hspace{2mm}\left.\frac{d}{dz_1}\left(Q(t_m^*(z''),(z_1,z_2))-Q(s_m^*(z'),(z_1,z_2))\right)\right|_{z_1=\tilde z_{1m}}\\
&\hspace{-0.5mm}\le&\hspace{-0.5mm}-\epsilon<0,
\end{eqnarray*}
where the second equality follows by letting $t_m^*(z)$ be a sequence in $N$ for which $Q(t_m^*(z),z)$ $\rightarrow V(N,z)$ as $m\rightarrow\infty$, and similarly for $s_m^*(z)$. The third equality follows by rearranging terms. The first inequality follows because $\lim_{m\rightarrow\infty}Q(t_m^*(z'),z')\le \underset{m\rightarrow\infty}{\text{limsup}}\hspace{1.5mm} Q(t_m^*(z''),z')$ and $\lim_{m\rightarrow\infty}Q(s_m^*(z''),z'')\le \underset{m\rightarrow\infty}{\text{limsup}}\hspace{1.5mm} Q(s_m^*(z'),z'')$. The final equality follows by the mean value theorem for some $\tilde z_{1m}$ between $z''_1$ and $z'_1$. 

The second inequality follows from (A.4) because $t_m^*(z')\in N$, $s_m^*(z'')\in M$, and $(\tilde z_{1m},z_2)\in W$, by convexity. This is a contradiction. Therefore, the set of all $z_1\in\mathbb{R}$ such that $V(N,z_1,z_2)=V(M,z_1,z_2)$ and $(z_1,z_2)\in W$ contains at most one point. Therefore, $P(\{z\in W|V(N,z)=V(M,z)\})=0$. 
\end{proof}

\section{Proof of Theorems and Corollary 1}

\begin{proof}[\bf{Proof of Theorem 1}]

For each $K=1,...,J$, let $T_K=\{(\tau,\mu): \tau\in\R^{K}, \mu\in\R^{K}, \tau_j>0, \sum_{j=1}^K\tau_j=1, \text{ and } \mu_1<\mu_2<...<\mu_K\}$. Notice that $Q(\tau,\mu,z)$ can be defined for any $(\mu,\tau)\in\cup_{K=1}^{J} T_K$ with $K$ in place of $J$. 

We verify the assumptions of Lemma 1. Assumption Manifold is satisfied because $T_J$ is a manifold in $\R^{2J}$ of dimension $2J-1$. Assumption Absolute Continuity is satisfied because $z_i$ is continuously distributed. We take $\mathcal{Z}$ to be the set of all $(z_1,...,z_n)\in\R^n$ that are distinct. This is permitted because any two draws of $z_i$ being equal occurs only on a set of measure zero. Assumption Continuous Differentiability is satisfied with $A(t)$ equal to the full tangent space because the objective function is continuously differentiable as a function of both $z$ and $(\tau,\mu)$. 

We verify a stronger version of Assumption Generic. For any $z\in\mathcal{Z}$ and for any $(\tau,\mu), (\varsigma,\nu)\in \cup_{K=1}^{J} T_K$, we show that if 
\begin{equation}
\frac{d}{dz_i}\left[Q(\tau,\mu,z_1,...,z_n)-Q(\varsigma,\nu,z_1,...,z_n)\right]=0
\end{equation}
for all $i=1,...,n$, then $\tau=\varsigma$ and $\mu=\nu$. By the contrapositive, if $\tau\neq \varsigma$ or $\mu\neq \nu$, then Condition (d) in Assumption Generic must hold for some $i=1,...,n$. If (B.1) holds for $\cup_{K=1}^J T_K$, then it also must hold for $T_J$. In that case, by Lemma 1, the argmin of $Q(\tau,\mu,z_1,...,z_n)$ is unique almost surely.

Let $(\tau,\mu)\in T_{K_1}$ and let $(\varsigma,\nu)\in T_{K_2}$. Suppose for every $i=1,...n$, the following holds: 
\begin{align*}
0&=\frac{d}{dz_i}[Q(\tau,\mu,z_1,...,z_n)-Q(\varsigma,\nu,z_1,...,z_n)]\\
&=\frac{1}{f(z_i;\tau,\mu)}\sum_{k=1}^{K_1} \tau_k (z_i-\mu_k)\phi(z_i-\mu_k)-\frac{1}{f(z_i;\varsigma,\nu)}\sum_{j=1}^{K_2} \varsigma_j (z_i-\nu_j)\phi(z_i-\nu_j).
\end{align*}
This holds if and only if 
\begin{align}
&&f(z_i;\varsigma,\nu)\sum_{k=1}^{K_1} \tau_k (z_i\hspace{-0.8mm}-\hspace{-0.8mm}\mu_k)\phi(z_i\hspace{-0.8mm}-\hspace{-0.8mm}\mu_k)&\hspace{-0.6mm}=\hspace{-0.6mm}f(z_i;\tau,\mu)\sum_{j=1}^{K_2} \varsigma_j (z_i\hspace{-0.8mm}-\hspace{-0.8mm}\nu_j)\phi(z_i\hspace{-0.8mm}-\hspace{-0.8mm}\nu_j)\nonumber\\
\iff&&\hspace{-2.8mm}\sum_{j=1}^{K_2} \varsigma_j \phi(z_i\hspace{-0.8mm}-\hspace{-0.8mm}\nu_j)\sum_{k=1}^{K_1} \tau_k (z_i\hspace{-0.8mm}-\hspace{-0.8mm}\mu_k)\phi(z_i\hspace{-0.8mm}-\hspace{-0.8mm}\mu_k)&\hspace{-0.6mm}=\hspace{-0.6mm}\sum_{k=1}^{K_1} \tau_k \phi(z_i\hspace{-0.8mm}-\hspace{-0.8mm}\mu_k)\sum_{j=1}^{K_2} \varsigma_j (z_i\hspace{-0.8mm}-\hspace{-0.8mm}\nu_j)\phi(z_i\hspace{-0.8mm}-\hspace{-0.8mm}\nu_j)\nonumber\\
\iff&&\hspace{-2.8mm}\sum_{j=1}^{K_2}\sum_{k=1}^{K_1} \varsigma_j  \tau_k\phi(z_i\hspace{-0.8mm}-\hspace{-0.8mm}\nu_j)\phi(z_i\hspace{-0.8mm}-\hspace{-0.8mm}\mu_k) (\mu_k\hspace{-0.8mm}-\hspace{-0.8mm}\nu_j)&\hspace{-0.6mm}=\hspace{-0.5mm}0.
\end{align}

Rewrite equation (B.2), for all $i=1,...,n$, as 
\begin{equation}
Ab=0,
\end{equation}
where $b$ is a $K_1 K_2\times 1$ vector, which is indexed by $j=1,...,K_1$ and $k=1,...,K_2$. The $(j,k)^{th}$ element of $b$ is given by: 
\[
\varsigma_j \tau_k (\mu_k-\nu_j)\exp\{-0.5(\nu_j^2+\mu_k^2)\}.
\]
$A$ is an $n\times K_1 K_2$ matrix, whose $\left(i, (j,k)\right)^{th}$ element is $\exp\{z_i(\nu_j+\mu_k)\}$, which uses 
\[
\phi(z_i-\nu_j)\phi(z_i-\mu_k)=(2\pi)^{-1}\exp\{-z_i^2\}\exp\{z_i(\nu_j+\mu_k)\}\exp\{-0.5(\nu_j^2+\mu_k^2)\}, 
\]
together with multiplying equation (B.2) by $2\pi \exp\{z_i^2\}$ to simplify. 

We would like to show that $A$ has full rank equal to $K_1K_2$, but this is not true because some of the columns of $A$ are redundant whenever there are pairs of indices, $(j_1,k_1)$ and $(j_2,k_2)$ such that $\nu_{j_1}+\mu_{k_1}=\nu_{j_2}+\mu_{k_2}$. To deal with this, we define equivalence classes of indices and shrink the matrix, $A$. 

Consider sets of indices, $(j,k)$, divided into equivalence classes according to the equivalence relation: 
\[
(j_1,k_1)\sim (j_2,k_2)\iff \nu_{j_1}+\mu_{k_1}=\nu_{j_2}+\mu_{k_2}.
\]
Let the set of all equivalence classes be $G$. For every $g\in G$, let $|g|$ denote the value of $\nu_j+\mu_k$ for $(j,k)\in g$. Suppose there are $M\le K_1 K_2$ equivalence classes, enumerated by $g_1,...,g_M$, in ascending order. Consider the following equation: 
\begin{equation}
\tilde A \tilde b =0,
\end{equation}
where $\tilde b$ is an $M$ vector, indexed by $g\in G$. The $g^{th}$ element of $\tilde b$ is given by: 
\[
\sum_{(j,k)\in g} \varsigma_j \tau_k (\mu_k-\nu_j)\exp\{-0.5(\nu_j^2+\mu_k^2)\}. 
\]
$\tilde A$ is a $n\times M$ matrix, whose $(i,g)^{th}$ element is: $\exp\{z_i|g|\}.$ We can see that equation (B.4) follows from equation (B.3) by eliminating those columns that are redundant. 

Next, we use the fact that $\tilde A$ is a strictly totally positive matrix (see \cite{Karlin1968}, 
Section 1.2). This implies that for distinct, $|g_1|,...,|g_M|$, and distinct $z_1,...,z_n$, $\tilde A$ has full rank $M$, using the fact that $n\ge J^2\ge K_1K_2\ge M$. Therefore, by inversion, we know that $\tilde b=0$ or, for every $g\in G$, 
\begin{equation}
\sum_{(j,k)\in g} \varsigma_j \tau_k(\mu_k-\nu_j) \exp\{-0.5(\nu_j^2+\mu_k^2)\}=0. 
\end{equation}
We use equation (B.5) to show the following Claim. \\

\noindent\textbf{Claim: } $K_1=K_2$ and for all $j,k=1,...,K_1$, $\mu_j=\nu_j$ and $\tau_k\varsigma_j=\tau_j\varsigma_k$. \\

\noindent\textbf{Proof of Claim: } The proof proceeds by induction on $j$ and $k$. 
\begin{enumerate}
\item Initialization Step: For $j=k=1$, $g_1=\{(1,1)\}$. This is because $\nu_j+\mu_k$ is minimized when both $\nu_j$ and $\mu_k$ are at their minimum values, $\nu_1$ and $\mu_1$. Applying equation (B.5) to $g_1$ gives
\[
\varsigma_1 \tau_1 (\mu_1-\nu_1)\exp\{-0.5(\nu_1^2+\mu_1^2)\}=0. 
\]
Since $\varsigma_1>0$ and $\tau_1>0$, this reduces to $\mu_1=\nu_1$. Also, $\tau_1\varsigma_1=\tau_1\varsigma_1$ holds trivially. 
\item Induction Step: For $j,k\le K$, assume that $\mu_j=\nu_j$ and $\tau_k\varsigma_j=\tau_j\varsigma_k$. We want to show the result for $j,k\le K+1$. 

Consider $\nu_{K+1}$ and $\mu_{K+1}$. Without loss of generality, suppose $\nu_{K+1}\le \mu_{K+1}$. Let $\bar g\in G$ be the smallest equivalence class that contains $K+1$ as an index. One can verify that $(K+1,1)\in \bar g$ and $|\bar g|=\nu_{K+1}+\mu_1$, since any other pair of indices containing $K+1$ must be larger than or equal to $\nu_{K+1}+\mu_1$. $\bar g$ may also contain three other types of indices: 
\begin{enumerate}
\item $(1,K+1)$ may belong to $\bar g$ depending on whether or not $\mu_{K+1}=\nu_{K+1}$. 
\item $(j,j)$ may belong to $\bar g$ for some $1\le j\le K$. 
\item The pair, $(j,k)$ and $(k,j)$, may both belong to $\bar g$, for $1\le j,k\le K$.  
\end{enumerate}
The proof that no other pairs of indices belongs to $\bar g$ follows from noticing the following facts. First, notice that if $j$ or $k$ is greater than $K+1$, then $\nu_j+\mu_k>\nu_{K+1}+\mu_1$. Second, notice that if either $j=K+1$ or $k=K+1$, then the other index must be equal to 1. Finally, notice that if $(j,k)\in \bar g$, where $j,k\le K$, then $(k,j)\in \bar g$ because, by the induction assumption, $\nu_j=\mu_j$ and $\nu_k=\mu_k$. 

Now, consider equation (B.5) applied to $\bar g$. We have that 
\begin{align}
0\hspace{1mm}&=\hspace{3.8mm}\varsigma_{K+1} \tau_1(\mu_1-\nu_{K+1}) \exp\{-0.5(\nu_{K+1}^2+\mu_1^2)\}\\
&\hspace{4mm}+\mathds{1}\{(1,K+1)\in\bar g\}\varsigma_{1} \tau_{K+1}(\mu_{K+1}-\nu_{1}) \exp\{-0.5(\nu_{1}^2+\mu_{K+1}^2)\}\nonumber\\
&\hspace{4mm}+\sum_{j=1}^K\mathds{1}\{(j,j)\in\bar g\}\varsigma_{j} \tau_{j}(\mu_{j}-\nu_{j}) \exp\{-0.5(\nu_{j}^2+\mu_{j}^2)\}\nonumber\\
&\hspace{4mm}+\sum_{j=2}^K\sum_{k=1}^{j-1}\mathds{1}\{(j,k)\in\bar g\}\left[\varsigma_{j} \tau_{k}(\mu_{k}-\nu_{j}) \exp\{-0.5(\nu_{j}^2+\mu_{k}^2)\}\right.\nonumber\\
&\hspace{4.5cm}\left.+\hspace{0.5mm}\varsigma_{k} \tau_{j}(\mu_{j}-\nu_{k}) \exp\{-0.5(\nu_{k}^2+\mu_{j}^2)\}\right]\hspace{-0.5mm}.\nonumber
\end{align}
The terms in the third line are zero because $\mu_j=\nu_j$ by the induction assumption. The terms in the fourth line are also zero because $\mu_k=\nu_k$, $\mu_j=\nu_j$, and $\varsigma_j\tau_k=\varsigma_k\tau_j$. Finally, there are two cases. If $(1,K+1)\notin \bar g$, then the equation 
\[
0=\varsigma_{K+1} \tau_1(\mu_1-\nu_{K+1}) \exp\{-0.5(\nu_{K+1}^2+\mu_1^2)\}
\]
implies a contradiction because $\varsigma_{K+1}>0$, $\tau_1>0$, $\nu_{K+1}>\nu_1=\mu_1$, and $\exp\{-0.5(\nu_{K+1}^2+\mu_1^2)\}>0$. Therefore, it must be the case that $(1,K+1)\in \bar g$. This can only happen if $\mu_{K+1}=\nu_{K+1}$. Further, in this case equation (B.5) reduces to: $\varsigma_{K+1}\tau_1=\varsigma_1\tau_{K+1}$. This proves the induction step because for any $j\le K$, we multiply the equation by $\frac{\tau_j}{\tau_1}=\frac{\varsigma_j}{\varsigma_1}$ to get: $\varsigma_{K+1}\tau_j=\varsigma_j\tau_{K+1}$.

\item The induction proceeds until $K=\min(K_1, K_2)$. We show that $K_1=K_2$. To reach a contradiction, suppose $K_1< K_2$. The case $K_2<K_1$ is treated symmetrically. Let $\bar g\in G$ be the smallest equivalence class that contains $K_1+1$ as an index. Then $(K_1+1,1)\in \bar g$ and $|\bar g|=\nu_{K_1+1}+\mu_1$. Following the argument of the induction step, equation (B.5) applied to $\bar g$ gives equation (B.6), except with the second line omitted. In this case, equation (B.6) becomes 
\[
0=\varsigma_{K_1+1} \tau_1(\mu_1-\nu_{K_1+1}) \exp\{-0.5(\nu_{K_1+1}^2+\mu_1^2)\},
\]
which is a contradiction because $\varsigma_{K_1+1}>0$, $\tau_1>0$, $\nu_{K_1+1}>\nu_1=\mu_1$, and $\exp\{-0.5(\nu_{K_1+1}^2+\mu_1^2)\}>0$. Therefore, it must be the case that $K_1=K_2$. 
\qed
\end{enumerate}

Finally, we show that $\varsigma=\tau$. From the Claim, we know that for any $j,k=1,...,K_1$, $\tau_k/\tau_j=\varsigma_k/\varsigma_j$. If $\tau_k>\varsigma_k$ for some $k$, then $\tau_j>\varsigma_j$ for any other $j$, to make the equation hold. Summing up, we get $1=\sum_{j=1}^{K_1} \tau_j>\sum_{j=1}^{K_1} \varsigma_j=1$, a contradiction. Therefore, $\varsigma_k=\tau_k$ for all $k$. This verifies (B.1). 
\end{proof}

\begin{proof}[\bf{Proof of Corollary 1}]

We apply Lemma 1 to $Q(\tau,\mu,z)$ defined over $T=\cup_{K=1}^J T_K$, where $T_K$ is defined at the beginning of the proof of Theorem 1. Assumption Manifold is satisfied because $T$ is a disjoint union of $T_K$, a manifold in $\R^{2K}$ of dimension $2K-1$. Assumption Absolute Continuity is satisfied because $z_i$ is continuously distributed. We take $\mathcal{Z}$ to be the set of all $(z_1,...,z_n)\in\R^n$ that are distinct. This is permitted because any two draws of $z_i$ being equal occurs only on a set of measure zero. Assumption Continuous Differentiability is satisfied with $A(t)$ equal to the full tangent space because the objective function is continuously differentiable, as a function of both $z$ and $(\tau,\mu)$. Assumption Generic is satisfied by (B.1) in the proof of Theorem 1. Therefore, $Q(\tau,\mu,z)$ is uniquely minimized over $T$ almost surely. 

Next, let $z$ be such that $Q(\tau,\mu,z)$ is uniquely minimized over $T$. Let $(\tau^*,\mu^*)$ be a global minimizer of $Q(\tau, \mu, z)$ over the unrestricted parameter space. Suppose $\nu^*$ is another $J$-vector of means and $\varsigma^*$ is another $J$-vector of positive weights that sum to $1$ satisfying $\sum_{j=1}^J\varsigma_j^*\mathds{1}\{\nu^*_j=\mu^*_k\} =\sum_{j=1}^J\tau_j^*\mathds{1}\{\mu^*_j=\mu^*_k\}$ for all $k$. Then, for all $z$, 
\[
f(z;\varsigma^*,\nu^*)=\sum_{j=1}^J \varsigma^*_j\phi(z-\nu^*_j)=\sum_{j=1}^J \tau^*_j\phi(z-\mu^*_j)=f(z;\tau^*,\mu^*). 
\]
This implies that $(\varsigma^*,\nu^*)$ is also a global minimizer of $Q(\tau,\mu,z)$ over the unrestricted parameter space. 

For the converse, let $(\varsigma^*,\nu^*)$ satisfy $\sum_{j=1}^J\varsigma_j^*\mathds{1}\{\nu^*_j=\mu^*_k\} \neq\sum_{j=1}^J\tau_j^*\mathds{1}\{\mu^*_j=\mu^*_k\}$ for some $k$. Let $(\varsigma^\dagger,\nu^\dagger)$ be a permutation and compression of $(\varsigma^*,\nu^*)$ so that $\nu^\dagger_1<\nu^\dagger_2<...<\nu^\dagger_{K_1}$ are the distinct elements of $\nu^*$ and $\varsigma^\dagger_k=\sum_{j=1}^J\varsigma^*_j\mathds{1}\{\nu^*_j=\nu^\dagger_k\}$. Similarly, let $(\tau^\dagger,\mu^\dagger)$ be a permutation and compression of $(\tau^*,\mu^*)$. Notice that $(\tau^\dagger,\mu^\dagger)$ and $(\varsigma^\dagger,\nu^\dagger)$ belong to $T$. Also notice that $Q(\tau^*,\mu^*,z)=Q(\tau^\dagger,\mu^\dagger,z)$ and $Q(\varsigma^*,\nu^*,z)=Q(\varsigma^\dagger,\nu^\dagger,z)$. This, together with the fact that $(\tau^*,\mu^*)$ is a global minimizer of $Q(\tau,\mu,z)$ over the unrestricted parameter space, implies that $(\tau^\dagger,\mu^\dagger)$ is a global minimizer of $Q(\tau,\mu,z)$ over $T$ (because each element of $T$ can be associated with some element of the unrestricted parameter space). By uniqueness, $(\tau^\dagger,\mu^\dagger)$ must be the unique global minimizer of $Q(\tau,\mu,z)$ over $T$. Also, the fact that there exists a $k$ such that $\sum_{j=1}^J\varsigma_j^*\mathds{1}\{\nu^*_j=\mu^*_k\} \neq\sum_{j=1}^J\tau_j^*\mathds{1}\{\mu^*_j=\mu^*_k\}$ implies that $(\tau^\dagger,\mu^\dagger)\neq(\varsigma^\dagger,\nu^\dagger)$. Therefore, $Q(\tau^\dagger,\mu^\dagger,z)<Q(\varsigma^\dagger,\nu^\dagger,z)$, so that $(\varsigma^*,\nu^*)$ is not a global minimizer of $Q(\tau,\mu,z)$ over the unrestricted parameter space. 
\end{proof}

\begin{proof}[\bf{Proof of Theorem 2}]

We verify the conditions of Lemma 1. Assumption Absolute Continuity is satisfied for $Y$ conditional on $X$. Assumption Manifold is satisfied by the assumption on $B$. Assumption Continuous Differentiability is satisfied because $p(\beta)$ is continuous on each $B_j$, $Q(\beta,Y,X)$ is continuously differentiable with respect to $Y$, and setting $A(t)=\{0\}$, so that no differentiability with respect to $\beta$ is needed. Consider condition (d) in Assumption Generic for $\beta_1\neq \beta_2$. If equation (d) in Assumption Generic is not satisfied for all $i=1,...,n$, then 
\[
\frac{d}{dY_i}\left(Q(\beta_1,Y,X)-Q(\beta_2,Y,X)\right)=X'_i(\beta_2-\beta_1)=0. 
\]
Stacking these up, we get that $X(\beta_2-\beta_1)=0$, and by the full rank of $X$ this implies $\beta_2=\beta_1$, a contradiction. Therefore, equation (d) in Assumption Generic is satisfied for some $i=1,...,n$. 
\end{proof}

\begin{proof}[\bf{Proof of Theorem 3}]

If, for every $M>0$, the argmin of $Q(t,W(\cdot))$ over $t\in[-M,M]$ is unique almost surely, then the argmin of $Q(t, W(\cdot))$ over $t\in\R$ is unique almost surely. This is because, if $\omega\in\Omega$ is the underlying state space, then 
\begin{align*}
&\{\omega\in\Omega| Q(t,W_\omega(\cdot)) \text{ has multiple minimizers over } \R\}\\
=&\cup_{M=1}^\infty \{\omega\in\Omega| Q(t,W_\omega(\cdot)) \text{ has multiple minimizers over } [-M,M]\},
\end{align*}
where $W_\omega(\cdot)$ denotes the draw of $W(\cdot)$ that is associated with $\omega\in\Omega$. Thus, we fix $M>0$ and show that the probability of multiple minimizers over $[-M,M]$ is zero.  

Let 
\[
W(t)=m(t)+B(t)+\Sigma_{t,M}\Sigma^{-1}_{M,M} Z,
\]
where $Z=W(M)-m(M)$ and $B(t)=W(t)-m(t)-\Sigma_{t,M}\Sigma^{-1}_{M,M}Z$ is independent of $Z$. Notice that the conditional distribution of $Z$ given $B(\cdot)$ is normal with variance $\Sigma_{M,M}>0$, and therefore continuously distributed. We condition on the realization of $B(\cdot)$ and verify Assumption Generic by taking a derivative with respect to $Z$. 

Fix $t_1>t_2$ and calculate
\[
\frac{d}{dZ}\left(Q(t_1,W(\cdot))-Q(t_2,W(\cdot))\right)=\int_{t_2}^{t_1}\phi(W(y))\Sigma_{y,M}\Sigma^{-1}_{M,M}dy,
\]
by the bounded convergence theorem, where $\phi(\cdot)$ is the standard normal pdf. Since the integrand on the right hand side is continuous and positive, for any $t_1>t_2$, the integral is positive, verifying Assumption Generic. Therefore, by Lemma 1, $Q(t,W(\cdot))$ has a unique minimum over $t\in [-M,M]$ almost surely, which implies uniqueness over $t\in \mathbb{R}$ almost surely. 
\end{proof}

\begin{lemma}
Let $T=\cup_{j\in J}T_j$ be a finite or countable disjoint union of second-countable Hausdorff manifolds. Let $f: T\rightarrow \mathbb{R}^l$ be continuously differentiable. Let $u\in\mathbb{R}^l$ be a \underline{non-invertible value} for $f$ if there exist $t_1\neq t_2\in T$ such that $f(t_1)=f(t_2)=u$. Let $f_n: T\rightarrow \mathbb{R}^l$ be a sequence of continuously differentiable and injective functions, converging uniformly over compact sets to $f$. Then, the Lebesgue measure of $\{u\in\mathbb{R}^l: u\text{ is a non-invertible value of $f$}\}$ is zero. 
\end{lemma}

\begin{proof}[\bf{Proof of Lemma 7}]

For each $j\in J$, denote the dimension of $T_j$ by $d_j$. Notice that in order for $f_n$ to be continuously differentiable and injective, $d_j$ must be less than or equal to $l$ for all $j\in J$. Denote the derivative of $f(t)$ by $f'(t)$. 

First, we show that any non-invertible value, $u$, must be a critical value, in the sense that there exists a $t\in T$ such that $f(t)=u$ and $f'(t)$ has rank less than $l$. Let $u$ be a non-invertible value. Then, there exist $t_1\neq t_2\in T$ such that $f(t_1)=f(t_2)=u$. Let $t_1\in T_{j_1}$ and $t_2\in T_{j_2}$. Suppose, to reach a contradiction, that $f'(t_1)$ and $f'(t_2)$ both have rank $l$. Notice that this requires $d_{j_1}=d_{j_2}=l$. Then, by the inverse function theorem, there exist neighborhoods, $N_1$ and $N_2$, of $t_1$ and $t_2$, and a neighborhood, $M$, of $u$ such that $f$ restricted to $N_i$ is invertible as a function onto $M$, for $i=1,2$. We can take $N_1, N_2$, and $M$ to be compact and homeomorphic to a closed ball in $\mathbb{R}^l$. We can also take $N_1$ and $N_2$ to be disjoint. 

Let $\epsilon>0$ such that the ball of radius $\epsilon$ around $u$ is contained in $M$. Then, by uniform convergence on compact sets, there exists an $n$ such that for all $t\in N_1\cup N_2$, $\|f_n(t)-f(t)\|<\epsilon/2$. For each $i=1,2$, let $f_n(N_i)=M_i$ denote the image of $N_i$ under $f_n$, and let $f_n(\partial N_i)=\partial M_i$ denote the image of the boundary of $N_i$, $\partial N_i$, under $f_n$. We know that $f$ maps the boundary of $N_i$ to the boundary of $M$ because $f$ is a homeomorphism when restricted to $N_i$, for each $i=1,2$. 

We show that $u\in M_i$ for both $i=1,2$. Suppose not for some $i$. Then, let $\bar u=f_n(f^{-1}(u))$, where the inverse of $f$ is restricted to $N_i$. By the choice of $n$, $\|\bar u-u\|<\epsilon/2$, and $\bar u\in M_i$. On the line segment connecting $u$ and $\bar u$, there exists at least one point, $u^\dagger\in \partial M_i$. 
Let $t^\dagger=f_n^{-1}(u^\dagger)$, and we notice that 
\[
\epsilon/2>\|f_n(t^\dagger)-f(t^\dagger)\|\ge \|u-f(t^\dagger)\|-\|f_n(t^\dagger)-u\|>\epsilon/2,
\]
where the first inequality follows by the choice of $n$, the second inequality is the triangle inequality, and the third inequality follows because (1) $f(t^\dagger)\in\partial M$, which is at least $\epsilon$ away from $u$ by the choice of $\epsilon$, and (2) $f_n(t^\dagger)=u^\dagger$ is on the line segment connecting $u$ and $\bar u$, so $\|u^\dagger-u\|\le \|\bar u-u\|<\epsilon/2$. This contradiction shows that $u\in M_i$ for both $i=1,2$. 

The fact that $u\in M_1\cap M_2$ is a contradiction because then there exist $\bar t_1\in N_1$ and $\bar t_2\in N_2$ such that $f_n(\bar t_1)=f_n(\bar t_2)=u$, which is impossible because $f_n$ is injective and $N_1\cap N_2=\emptyset$. Altogether, this shows that any non-invertible value, $u$, must be a critical value. 

By Sard's Theorem (see \cite{GuilleminPollack1974}), the set of critical values in $\mathbb{R}^l$ has Lebesgue measure zero. 
\end{proof}

\begin{proof}[\bf{Proof of Theorem 4}]

We verify the conditions of Lemma 1. Assumption Manifold is satisfied by the assumption on the identified set. The fact that $z$ is a continuous random vector implies Assumption Absolute Continuity. Assumption Continuous Differentiability is satisfied by Assumption Weak Identification (a). The only assumption we need to verify is Assumption Generic. 

Let $\pi\in\Pi$ denote the identified set. Consider $z\in\mathbb{R}^{d_z}$ and write $z=(z_1,z_2)$, where $z_1\in\mathbb{R}^{d_\beta}$ and $z_2\in\mathbb{R}^{d_h}$. For a fixed $z=(z_1, z_2)$ consider the equation 
\begin{equation}
h_\beta(\beta_0,\pi) (z_1-b)=z_2-h_\beta(\beta_0,\pi_0)b. 
\end{equation}
We show that $\{z=(z_1,z_2): \text{ equation (B.7) has multiple solutions over } \pi\in\Pi\}$ has Lebesgue measure zero. 

Fix $z_1$, and for every $n$ let 
\[
\tilde\beta_n=\beta_n-\frac{z_1}{\sqrt{n}}+\frac{a_n}{n},
\]
where $\|a_n\|\le 1$ is chosen so that $h(\tilde\beta_n,\pi)-h(\beta_0,\pi)$ is injective as a function of $\pi$. Such an $a_n$ always exists by Assumption (d) and for $n$ large enough so that $\beta_n-z_1/\sqrt{n}\in B$. Let 
\begin{align*}
f(\pi)&=h_\beta(\beta_0,\pi)(z_1-b), \text{ and }\\
f_n(\pi)&=-\sqrt{n}\left(h(\tilde\beta_n,\pi)-h(\beta_0,\pi)\right). 
\end{align*}

We verify the conditions of Lemma 7. Let $K$ be a compact subset of $\Pi$, and notice 
\begin{align*}
f_n(\pi)&=-\sqrt{n}\left(h(\tilde\beta_n,\pi)-h(\beta_n,\pi)\right)-\sqrt{n}\left(h(\beta_n,\pi)-h(\beta_0,\pi)\right)\\
&=-h_\beta(\ddot\beta_{n, \pi},\pi)\left(-z_1+\frac{a_n}{\sqrt{n}}\right)-\sqrt{n}\left(h(\beta_n,\pi)-h(\beta_0,\pi)\right)\\
&\rightarrow h_\beta(\beta_0,\pi)z_1-h_\beta(\beta_0,\pi)b, \text{ uniformly over $\pi\in K$}\\
&=f(\pi),
\end{align*}
where the second equality follows by a mean value expansion for every $\pi\in K$ and $\ddot \beta_{n,\pi}$ is on the line segment between $\tilde\beta_n$ and $\beta_n$. The convergence follows because (1) the second term converges to $h_\beta(\beta_0,\pi)b$ uniformly over $\pi\in K$ by Assumption Weak Identification (b), and (2) the first term converges to $h_\beta(\beta_0,\pi)z_1$ uniformly over $\pi\in K$ by the continuity of $h_\beta(\beta,\pi)$ and the fact that $a_n$ is bounded. 

Therefore, by Lemma 7, the set of all $z_2\in\mathbb{R}^{d_h}$ such that the equation $f(\pi)=z_2-h_\beta(\beta_0,\pi_0)b$ has multiple solutions over $\pi\in\Pi$ has Lebesgue measure zero. Since this holds for every $z_1$, the Lebesgue measure of $\{z=(z_1,z_2):$ equation (B.7) has multiple solutions over $\pi\in\Pi\}$ has Lebesgue measure zero. 

Next, let $\mathcal{Z}=\{z\in\mathbb{R}^{d_z}: \text{ equation (B.7) admits at most one solution}\}$. We want to show that $\xi(\pi_1,\pi_2,z)$ satisfies Assumption Generic on $\Xi=\{(\pi_1,\pi_2,z): \pi_1,\pi_2\in \Pi, \pi_1\neq \pi_2, z\in \mathcal{Z}\}$. We write $\xi(\pi_1,\pi_2,z)$ as 
\begin{align*}
&\hspace{5mm}Q(\pi_1,z)-Q(\pi_2,z)\\
&=2z'{H^{1/2}}'(g(\pi_1)-g(\pi_2))+\kappa(\pi_1)-\kappa(\pi_2)\\
&\hphantom{==}-(H^{1/2}z+g(\pi_1))'P(\pi_1)(H^{1/2}z+g(\pi_1))+(H^{1/2}z+g(\pi_2))'P(\pi_2)(H^{1/2}z+g(\pi_2))\\
&=2W'(2I_{d_z}-P(\pi_1)-P(\pi_2))\delta-W'(P(\pi_1)-P(\pi_2))W+\tilde\kappa(\pi_1,\pi_2),
\end{align*}
where $W=H^{1/2}z+\frac{g(\pi_1)+g(\pi_2)}{2}$ and $\delta=\frac{g(\pi_1)-g(\pi_2)}{2}$, and the second equality follows from some algebra for some $\tilde \kappa(\pi_1,\pi_2)$ that does not depend on $z$. 

Suppose part (d) in Assumption Generic is not satisfied. Then, 
\begin{align*}
0&=\frac{d}{dz}\xi(\pi_1,\pi_2,z)=2{H^{1/2}}'(2I_{d_z}-P(\pi_1)-P(\pi_2))\delta-2{H^{1/2}}'(P(\pi_1)-P(\pi_2))W. 
\end{align*}
This implies that 
\begin{equation}
0=M(\pi_1)(W+\delta)+M(\pi_2)(\delta-W), 
\end{equation}
where $M(\pi)=I_{d_z}-P(\pi)$, which projects onto the span of ${H^{-1/2}}'\left[\begin{array}{c}-h_\beta(\beta_0,\pi)'\\I_{d_h}\end{array}\right]$. Condition (c) implies that the ranges of $M(\pi_1)$ and $M(\pi_2)$ are linearly independent. Then, equation (B.8) implies that both $M(\pi_1)(W+\delta)=0$ and $M(\pi_2)(\delta-W)=0$. Premultiply both equations by $\left[0_{d_\beta\times d_h}, I_{d_h}\right]H^{1/2}$ to get 
\begin{align*}
\left(\left[\begin{array}{c}-h_\beta(\beta_0,\pi_1)'\\I_{d_h}\end{array}\right]'\hspace{-1mm}H^{-1}\hspace{-1mm}\left[\begin{array}{c}-h_\beta(\beta_0,\pi_1)'\\I_{d_h}\end{array}\right]\hspace{-1mm}\right)^{-1}\left[\begin{array}{c}-h_\beta(\beta_0,\pi_1)'\\I_{d_h}\end{array}\right]'\hspace{-1mm}{H^{-1/2}}(W+\delta)&=0\\
\left(\left[\begin{array}{c}-h_\beta(\beta_0,\pi_2)'\\I_{d_h}\end{array}\right]'\hspace{-1mm}H^{-1}\hspace{-1mm}\left[\begin{array}{c}-h_\beta(\beta_0,\pi_2)'\\I_{d_h}\end{array}\right]\hspace{-1mm}\right)^{-1}\left[\begin{array}{c}-h_\beta(\beta_0,\pi_2)'\\I_{d_h}\end{array}\right]'\hspace{-1mm}{H^{-1/2}}(\delta-W)&=0. 
\end{align*}
Premultiplying by invertible matrices and using the formulas for $W$ and $\delta$ gives 
\begin{align*}
\left[\begin{array}{c}-h_\beta(\beta_0,\pi_1)'\\I_{d_h}\end{array}\right]'{H^{-1/2}}(H^{1/2}z+g(\pi_1))&=0\\
\left[\begin{array}{c}-h_\beta(\beta_0,\pi_2)'\\I_{d_h}\end{array}\right]'{H^{-1/2}}(H^{1/2}z+g(\pi_2))&=0. 
\end{align*}
Rearranging, using $z=(z_1, z_2)'$, and using the formula for $g(\pi)$ gives 
\begin{align*}
h_\beta(\beta_0,\pi_1)(z_1-b)&=z_2-h_\beta(\beta_0,\pi_0)b\\
h_\beta(\beta_0,\pi_2)(z_1-b)&=z_2-h_\beta(\beta_0,\pi_0)b. 
\end{align*}

This shows that for $(z_1,z_2)\in\mathcal{Z}$, both $\pi_1$ and $\pi_2$ are solutions to equation (B.7). This is a contradiction because we have assumed that for any $z\in\mathcal{Z}
$, equation (B.7) has at most one solution. Therefore, $\frac{d}{dz}\xi(\pi_1,\pi_2,z)\neq 0$. This must hold for all $(\pi_1,\pi_2,z)\in\Xi$, showing that $\xi(\pi_1,\pi_2,z)$ satisfies Assumption Generic. 
\end{proof}

\bibliography{uniqueminreferences}

\begin{thebibliography}{}

\bibitem[\protect\astroncite{Andrews and Cheng}{2012}]{AndrewsCheng2012}
Andrews, D. and Cheng, X. (2012).
\newblock Estimation and inference with weak, semi-strong, and strong
  identification.
\newblock {\em Econometrica}, 80:2153--2211.

\bibitem[\protect\astroncite{Arcones}{1992}]{Arcones1992}
Arcones, M. (1992).
\newblock On the arg max of a {G}aussian process.
\newblock {\em Statistics and Probability Letters}, 15:373--374.

\bibitem[\protect\astroncite{Cheng and Chen}{1988}]{ChengChen1988}
Cheng, K. and Chen, C. (1988).
\newblock Estimation of the weibull parameters with grouped data.
\newblock {\em Communications in Statistics - Theory and Methods}, 17:325--341.

\bibitem[\protect\astroncite{Cheng}{2015}]{Cheng2015}
Cheng, X. (2015).
\newblock Robust inference in nonlinear models with mixed identification
  strength.
\newblock {\em Journal of Econometrics}, 189:207--228.

\bibitem[\protect\astroncite{Copas}{1975}]{Copas1975}
Copas, J. (1975).
\newblock On the unimodality of the likelihood for the cauchy distribution.
\newblock {\em Biometrika}, 62:701--704.

\bibitem[\protect\astroncite{Cox}{2019}]{Cox2019}
Cox, G. (2019).
\newblock Weak identification with bounds in a class of minimum distance
  models.
\newblock Unpublished Manuscript.

\bibitem[\protect\astroncite{Demidenko}{2008}]{Demidenko2008}
Demidenko, E. (2008).
\newblock Criteria for unconstrained global optimization.
\newblock {\em Journal of Optimization Theory and Applications}, 136:375--395.

\bibitem[\protect\astroncite{Fan and Li}{2001}]{FanLi2001}
Fan, J. and Li, R. (2001).
\newblock Variable selection via nonconcave penalized likelihood and its oracle
  properties.
\newblock {\em Journal of the American Statistical Association}, 96:1348--1360.

\bibitem[\protect\astroncite{Fan and Peng}{2004}]{FanPeng2004}
Fan, J. and Peng, H. (2004).
\newblock Nonconcave penalized likelihood with a diverging number of
  parameters.
\newblock {\em The Annals of Statistics}, 32:928--961.

\bibitem[\protect\astroncite{Fan et~al.}{2014}]{FanXueZou2014}
Fan, J., Xue, L., and Zou, H. (2014).
\newblock Strong oracle optimality of folded concave penalized estimation.
\newblock {\em The Annals of Statistics}, 42:819--849.

\bibitem[\protect\astroncite{Ferger}{1999}]{Ferger1999}
Ferger, D. (1999).
\newblock On the uniqueness of maximizers of markov-gaussian processes.
\newblock {\em Statistics and Probability Letters}, 45:71--77.

\bibitem[\protect\astroncite{Guillemin and
  Pollack}{1974}]{GuilleminPollack1974}
Guillemin, V. and Pollack, A. (1974).
\newblock {\em Differential Topology}.
\newblock Prentice-Hall.

\bibitem[\protect\astroncite{Han and McCloskey}{2018}]{HanMcCloskey2018}
Han, S. and McCloskey, A. (2018).
\newblock Estimation and inference with a (nearly) singular jacobian.
\newblock Unpublished Manuscript.

\bibitem[\protect\astroncite{Hill et~al.}{1980}]{HillSaundersLaud1980}
Hill, D., Saunders, R., and Laud, P. (1980).
\newblock Maximum likelihood estimation for mixtures.
\newblock {\em The Canadian Journal of Statistics}, 8:87--93.

\bibitem[\protect\astroncite{Hillier and
  Armstrong}{1999}]{HillierArmstrong1999}
Hillier, G. and Armstrong, M. (1999).
\newblock The density of the maximum likelihood estimator.
\newblock {\em Econometrica}, 67:1459--1470.

\bibitem[\protect\astroncite{Huang et~al.}{2008}]{HuangHorowitzMa2008}
Huang, J., Horowitz, J., and Ma, S. (2008).
\newblock Asymptotic properties of bridge estimators in sparse high-dimensional
  regression models.
\newblock {\em The Annals of Statistics}, 36:587--613.

\bibitem[\protect\astroncite{Jewell}{1982}]{Jewell1982}
Jewell, N. (1982).
\newblock Mixtures of exponential distributions.
\newblock {\em The Annals of Statistics}, 10:479--484.

\bibitem[\protect\astroncite{Karlin}{1968}]{Karlin1968}
Karlin, S. (1968).
\newblock {\em Total Positivity}, volume~1.
\newblock Stanford University Press.

\bibitem[\protect\astroncite{Kim and Pollard}{1990}]{KimPollard1990}
Kim, J. and Pollard, D. (1990).
\newblock Cube root asymptotics.
\newblock {\em The Annals of Statistics}, 18:191--219.

\bibitem[\protect\astroncite{Kim et~al.}{2008}]{KimChoiOh2008}
Kim, Y., Choi, H., and Oh, H.-S. (2008).
\newblock Smoothly clipped absolute deviation on high dimensions.
\newblock {\em Journal of the American Statistical Association},
  103:1665--1673.

\bibitem[\protect\astroncite{Lifshits}{1982}]{Lifshits1982}
Lifshits, M. (1982).
\newblock On the absolute continuity of the distributions of functionals of
  stochastic processes.
\newblock {\em Theory of Probability and its Applications}, 27:600--607.

\bibitem[\protect\astroncite{Lindsay}{1983a}]{Lindsay1983a}
Lindsay, B. (1983a).
\newblock The geometry of mixture likelihoods: A general theory.
\newblock {\em The Annals of Statistics}, 11:86--94.

\bibitem[\protect\astroncite{Lindsay}{1983b}]{Lindsay1983b}
Lindsay, B. (1983b).
\newblock The geometry of mixture likelihoods, part ii: The exponential family.
\newblock {\em The Annals of Statistics}, 11:783--792.

\bibitem[\protect\astroncite{Lindsay et~al.}{1991}]{LindsayCloggGrego1991}
Lindsay, B., Clogg, C., and Grego, J. (1991).
\newblock Semiparametric estimation in the {R}asch model and related
  exponential response models, including a simple latent class model for item
  analysis.
\newblock {\em Journal of the American Statistical Association}, 86:96--107.

\bibitem[\protect\astroncite{Lindsay and Roeder}{1993}]{LindsayRoeder1993}
Lindsay, B. and Roeder, K. (1993).
\newblock Uniqueness of estimation and identifiability in mixture models.
\newblock {\em The Canadian Journal of Statistics}, 21:139--147.

\bibitem[\protect\astroncite{Loh and Wainwright}{2017}]{LohWainwright2017}
Loh, P. and Wainwright, M. (2017).
\newblock Support recovery without incoherence: A case for nonconvex
  regularization.
\newblock {\em The Annals of Statistics}, 45:2455--2482.

\bibitem[\protect\astroncite{L\'{o}pez and Pimentel}{2016}]{LopezPimentel2016}
L\'{o}pez, S. and Pimentel, L. (2016).
\newblock On the location of the maximum of a process: L\'{e}vy, gaussian, and
  multidimensional cases.
\newblock \textit{arXiv}, Preprint Number 1611:02334.

\bibitem[\protect\astroncite{M\"{a}kel\"{a}inen
  et~al.}{1981}]{MakelainenSchmidtStyan1981}
M\"{a}kel\"{a}inen, T., Schmidt, K., and Styan, G. (1981).
\newblock On the existence and uniqueness of the maximum likelihood estimate of
  a vector-valued parameter in fixed-size samples.
\newblock {\em The Annals of Statistics}, 9:758--767.

\bibitem[\protect\astroncite{Mallet}{1986}]{Mallet1986}
Mallet, A. (1986).
\newblock A maximum likelihood estimation method for random coefficient
  regression models.
\newblock {\em Biometrika}, 73:645--656.

\bibitem[\protect\astroncite{Mallik et~al.}{2013}]{MallikBanerjeeSen2013}
Mallik, A., Banerjee, M., and Sen, B. (2013).
\newblock Asymptotics for p-value based trheshold estimation in regression
  settings.
\newblock {\em Electronic Journal of Statistics}, 7:2477--2515.

\bibitem[\protect\astroncite{Mascarenhas}{2010}]{Mascarenhas2010}
Mascarenhas, W. (2010).
\newblock A mountain pass lemma and its implications regarding the uniqueness
  of constrained minimizers.
\newblock {\em Optimization}, 60:1121--1159.

\bibitem[\protect\astroncite{M\"{u}ller and Song}{1996}]{MullerSong1996}
M\"{u}ller, H. and Song, K. (1996).
\newblock A set-indexed process in a two-region image.
\newblock {\em Stochastic Processes and Their Applications}, 62:87--101.

\bibitem[\protect\astroncite{Olsen}{1978}]{Olsen1978}
Olsen, R. (1978).
\newblock Note on the uniqueness of the maximum likelihood estimator for the
  tobit model.
\newblock {\em Econometrica}, 46:1211--1215.

\bibitem[\protect\astroncite{Orme}{1989}]{Orme1989}
Orme, C. (1989).
\newblock On the uniqueness of the maximum likelihood estimator in truncated
  regression models.
\newblock {\em Econometric Reviews}, 8:217--222.

\bibitem[\protect\astroncite{Orme and Ruud}{2002}]{OrmeRuud2002}
Orme, C. and Ruud, P. (2002).
\newblock On the uniqueness of the maximum likelihood estimator.
\newblock {\em Economics Letters}, 75:209--217.

\bibitem[\protect\astroncite{Pimentel}{2014}]{Pimentel2014}
Pimentel, L. (2014).
\newblock On the location of the maximum of a continuous stochastic process.
\newblock {\em Journal of Applied Probability}, 51:152--161.

\bibitem[\protect\astroncite{Ro\'{s} et~al.}{2016}]{RosBijmaMunckGunst2016}
Ro\'{s}, B., Bijma, F., de~Munck, J., and de~Gunst, M. (2016).
\newblock Existence and uniqueness of the maximum likelihood estimator for
  models with a kronecker product covariance structure.
\newblock {\em Journal of Multivariate Analysis}, 143:345--361.

\bibitem[\protect\astroncite{Seregin}{2010}]{Seregin2010}
Seregin, A. (2010).
\newblock Uniqueness of the maximum likelihood estimator for k-monotone
  densities.
\newblock {\em Proceedings of the American Mathematical Society},
  138:4511--4515.

\bibitem[\protect\astroncite{Simar}{1976}]{Simar1976}
Simar, L. (1976).
\newblock Maximum likelihood estimation of a compound poisson process.
\newblock {\em The Annals of Statistics}, 4:1200--1209.

\bibitem[\protect\astroncite{Soloveychik and
  Trushin}{2016}]{SoloveychikTrushin2016}
Soloveychik, I. and Trushin, D. (2016).
\newblock {G}aussian and robust kronecker product covariance estimation:
  Existence and uniqueness.
\newblock {\em Journal of Multivariate Analysis}, 149:92--113.

\bibitem[\protect\astroncite{Stock and Wright}{2000}]{StockWright2000}
Stock, J. and Wright, J. (2000).
\newblock Gmm with weak identification.
\newblock {\em Econometrica}, 68:1055--1096.

\bibitem[\protect\astroncite{Tarone and Gruenhage}{1975}]{TaroneGruenhage1975}
Tarone, R. and Gruenhage, G. (1975).
\newblock A note on the uniqueness of roots of the likelihood equations for
  vector-valued parameters.
\newblock {\em Journal of the American Statistical Association}, 70:903--904.

\bibitem[\protect\astroncite{van~der Vaart and
  Wellner}{1996}]{VaartWellner1996}
van~der Vaart, A. and Wellner, J. (1996).
\newblock {\em Weak Convergence}.
\newblock Springer.

\bibitem[\protect\astroncite{Wall}{2016}]{Wall2016}
Wall, C. (2016).
\newblock {\em Differential Topology}.
\newblock Cambridge University Press.

\bibitem[\protect\astroncite{Wang and Bice}{1997}]{WangBice1997}
Wang, W. and Bice, D. (1997).
\newblock A model with mixed binary responses and censored observations.
\newblock {\em Communications in Statistics - Theory and Methods}, 26:921--941.

\bibitem[\protect\astroncite{Wood}{1999}]{Wood1999}
Wood, G. (1999).
\newblock Binomial mixtures: Geometric estimation of the mixing distribution.
\newblock {\em The Annals of Statistics}, 27:1706--1721.

\bibitem[\protect\astroncite{Zhang}{2010}]{Zhang2010}
Zhang, C. (2010).
\newblock Nearly unbiased variable selection under minimax concave penalty.
\newblock {\em The Annals of Statistics}, 38:894--942.

\bibitem[\protect\astroncite{Zhang and Zhang}{2012}]{ZhangZhang2012}
Zhang, C. and Zhang, T. (2012).
\newblock A general theory of concave regularization for high-dimensional
  sparse estimation problems.
\newblock {\em Statistical Science}, 27:576--593.

\end{thebibliography}

\end{document}